\documentclass{article}
\usepackage{amsmath}
\usepackage{amsfonts}
\usepackage{amssymb}
\usepackage{cases}
\usepackage{amsthm}

\usepackage{multirow}

\title{Weighted and locally bounded list-colorings in split graphs, cographs, and partial $k$-trees}
\author{C\'edric Bentz\footnote{CNAM \& CEDRIC Laboratory, 292 rue Saint-Martin 75003 Paris, France (cedric.bentz@cnam.fr)}}
\date{}

\newtheorem{Th}{Theorem}
\newtheorem{Lm}{Lemma}
\newtheorem{Cor}{Corollary}
\newtheorem{Proposition}{Proposition}

\begin{document}
 \maketitle

\begin{abstract}
For a fixed number of colors, we show that, in node-weighted split graphs, cographs, and graphs of bounded tree-width, one can determine in polynomial time whether a proper list-coloring of the vertices of a graph such that the total weight of vertices of each color equals a given value in each part of a fixed partition of the vertices exists. We also show that this result is tight in some sense, by providing hardness results for the cases where any one of the assumptions does not hold. The edge-coloring variant is also studied, as well as special cases of cographs and split graphs.
\end{abstract}

\section{Introduction}\label{sec:Intro}

A \emph{proper coloring} of a given graph $G$ is an assignment of colors (integers) to its vertices, such that any two vertices linked by an edge of $G$ take different colors. For any given color, the set of vertices taking this color is called a \emph{color class}. In \cite{refBentz09}, a coloring problem with non global constraints on the sizes of the color classes was studied. More precisely, the following problem was considered:\\

\noindent {\sc LocallyBoundedColoring}\\
\textit{Instance}: A graph $G=(V,E)$, a partition $V_1, \dots, V_p$ of the vertex set $V$, and a list of $pk$ integral bounds $(n_{11}, \dots,
n_{1k}, n_{21}, \dots, n_{2k}, \dots, n_{p1}, \dots, n_{pk})$ such that $\sum_{j=1}^{k} n_{ij} = \vert V_i \vert$ for each $i \in \{1, \dots, p\}$.\\
\textit{Question}: Decide whether there exists a proper $k$-coloring of $G$ (i.e., a proper coloring of $G$ using $k$ colors) that is such that, for each $i \in \{1, \dots, p\}$ and for each $j \in \{1, \dots, k\}$, the number of vertices having color $j$ in $V_i$ is $n_{ij}$.\\

In this previous paper, it was shown that this problem can be solved in $O(n^{2pk-1}\log n)$ time when $G$ is a tree, where $n=\vert V \vert$. A link with a scheduling problem consisting in processing a set of unit tasks on a set of processors with various unavailability constraints was also presented.

Moreover, it was shown in \cite{refGravier02} that the following \emph{list-coloring} problem (i.e., where each vertex must take a color from a list of possible colors) with \emph{global} constraints on the size of the color classes is tractable in several classes of graphs:\\

\noindent {\sc BoundedListColoring}\\
\textit{Instance}: A graph $G=(V,E)$, a list of $k$ integral bounds $(n_{1}, \dots,
n_{k})$, and, for each vertex $v$ of $G$, a list of possible colors $L(v) \subseteq \{1, \dots, k\}$.\\
\textit{Question}: Decide whether there exists a proper $k$-coloring of $G$ such that, for each $i \in \{1, \dots, k\}$, the number of vertices having color $i$ is at most equal to the \emph{color bound} $n_i$, and the color of $v$ belongs to $L(v)$ for each vertex $v$.\\

More precisely, the authors of \cite{refGravier02} described efficient algorithms to solve this problem in graphs of bounded tree-width and in cographs (that can be defined as graphs with no induced path on four vertices \cite{refBrandstadt}). The notion of tree-width will be recalled formally when needed: it can be viewed as a measure of the tree-likeness of a graph, and equals 1 in forests (and 0 in graphs with no edges). Graphs of tree-width at most $t$ for some $t \geq 0$ are also called \emph{partial $t$-trees} \cite{refBrandstadt}.

In this paper, we generalize both results, and show how to solve efficiently the list-coloring problem with non global and weighted constraints on the sizes of the color classes in graphs of bounded tree-width (and also in cographs). More precisely, the problem we shall study can be formally defined as follows:\\

\noindent {\sc WeightedLocallyBoundedListColoring}\\
\textit{Instance}: A graph $G=(V,E)$, a weight function $w: V \rightarrow \mathbf N^*$, a partition $V_1, \dots, V_p$ of the vertex set $V$, a list of $pk$ integral bounds $(W_{11}, \dots, W_{1k}, W_{21}, \dots,$ $W_{2k}, \dots, W_{p1}, \dots, W_{pk})$ with $\sum_{j=1}^{k} W_{ij} = \sum_{v:v \in V_i} w(v)$ for each $i \in \{1, \dots, p\}$, and, for each vertex $v$ of $G$, a list of possible colors $L(v) \subseteq \{1, \dots, k\}$.\\
\textit{Question}: Decide whether there exists a proper $k$-coloring of $G$ such that, for each $i \in \{1, \dots, p\}$ and for each $j \in \{1, \dots, k\}$, the total weight of vertices having color $j$ in $V_i$ is $W_{ij}$, and the color of $v$ belongs to $L(v)$ for each vertex $v$.\\

In the scheduling application mentioned above, the weight of each vertex would correspond to the amount of resource needed to process the associated task. Moreover, apart from list-coloring problems, other constrained coloring problems related to ours have been studied in the literature, see for instance \cite{refBaker96,refBodlaender05,refBodlaenderJansen95,refBonomo11,refGravier99,refGolovach11,refHertz93,refJansen}. We will not review all the associated known results, as we will not need them all, and will only recall the useful ones when needed.

The \emph{equitable} coloring problem consists in finding a coloring where the sizes of any two color classes differ by at most one, while the \emph{bounded} coloring problem consists in finding a coloring where the size of any color class does not exceed a given bound (that is common to all color classes). Moreover, without the requirement of being a list-coloring, the problem studied in \cite{refGravier02} is defined as the \emph{capacitated} coloring problem in \cite{refBonomo11}. Hence, the bounded coloring problem is a special case of the capacitated coloring problem, where all the color bounds (i.e., all the bounds on the sizes of the color classes) are equal.

It was observed in \cite{refBonomo11} that the equitable coloring problem is in fact a special case of the capacitated coloring problem (and one in which the color bounds are actually reached), by setting $n_i = \left\lceil\frac{n}{k}\right\rceil$ for any $i \leq n \mod k$, and $n_i = \left\lfloor\frac{n}{k}\right\rfloor$ for any other $i$, where $n$ is the number of vertices. It was also observed in \cite{refBodlaender05} that any instance of the bounded coloring problem is equivalent to an instance of the equitable coloring problem with the same number of colors, by adding sufficiently many isolated vertices (which destroys neither the property of having bounded tree-width nor the one of being a cograph or a split graph, i.e., a graph whose vertex set can be partitioned into a clique and an independent set \cite{refBrandstadt}).

In fact, all of the above problems (including the $k$-coloring problem, which is a list-coloring problem where $L(v)=\{1, \dots, k\}$ for each vertex $v$) are special cases of {\sc BoundedListColoring}, and this is true for {\sc WeightedLocallyBoundedListColoring} as well, in most of the graphs that we shall consider: take any instance of {\sc BoundedListColoring},
keep the same color bounds, set $p$ and the weight of any vertex to 1, and add isolated vertices of weight 1 (that can take any color), so that the total number of vertices reaches $\sum_{i=1}^k n_i$.

Numerous applications of the list-coloring problem and of the capacitated coloring problem were described respectively in \cite{refJansen} and in \cite{refBodlaenderJansen95,refBonomo11}. Many papers also consider a precolored variant of the equitable or capacitated coloring problem, where some of the vertices are already colored \cite{refBodlaender05}. Observe that this can also be achieved, when considering an instance of the list-coloring problem, by defining the list of possible colors for each vertex that is already colored as a singleton, corresponding to the given color. Similarly, one can transform a list-coloring instance into a precolored one, by adding pendant vertices adjacent to each vertex $v$ (these pendant vertices being colored by the colors not in $L(v)$). Such an operation does not destroy the property of having bounded tree-width, but can destroy the one of being a cograph (or a split graph, depending on whether $v$ lies in the clique or in the independent set). Finally, when considering instances of {\sc WeightedLocallyBoundedListColoring} where $p$ is not fixed and $L(v) = \{1, \dots, k\}$ for each vertex $v$, then precoloring some vertices can be achieved by putting each of these vertices in one set $V_i$ of the partition $V_1, \dots, V_p$ ($V_i$ contains only this vertex), and defining the $W_{ij}$'s appropriately.

The present paper is organized as follows. We begin by providing, in Section \ref{sect:NPC}, tight hardness results for {\sc WeightedLocallyBoundedListColoring}. Then, we describe in Section \ref{sect:btw} a dynamic programming algorithm for this problem, that runs in polynomial time whenever we are not in the cases covered by Section \ref{sect:NPC} (i.e., when $p$, $k$ and the tree-width of the graph are fixed, and the vertex weights are polynomially bounded). In Section \ref{sect:cographs}, we prove that the problem is tractable in cographs under similar assumptions, and that, under weaker assumptions, it can be solved in polynomial time in several subclasses of cographs. Section \ref{sect:split-graphs} is devoted to split graphs, and we show that, in such graphs, the problem can be solved in (pseudo)polynomial time even when $p$ is not fixed (provided that $k$ is), and provide additional tractable special cases. Finally, in the last section of the paper, we extend to edge colorings all our results concerning (vertex) colorings in the graphs studied in Sections \ref{sect:btw} to \ref{sect:split-graphs}.

\section{\textbf{NP}-completeness proofs}\label{sect:NPC}

In this section, we prove that {\sc WeightedLocallyBoundedListColoring}, which is clearly in \textbf{NP}, is \textbf{NP}-complete, even in very restricted special cases.

Before detailing these special cases, we note that, when $k=3$, {\sc WeightedLocallyBoundedListColoring} is \textbf{NP}-complete in general graphs (i.e., when the tree-width is unbounded), even when $p=1$ and $w(v)=1$ for each vertex $v$, as deciding whether a graph can be colored using three colors is. Also note that this latter result does not hold when $k=2$: in this case, the coloring problem is trivial, as the input graph is bipartite (otherwise, the answer is \emph{no}).

Actually, {\sc WeightedLocallyBoundedListColoring} itself can be solved in polynomial time when $k=2$ and $p = O(1)$, provided that the vertex weights are polynomially bounded (in the next subsection, we deal with the case where they are not). Indeed, if $k=2$, then any connected component of the graph is bipartite (otherwise, the answer is \emph{no}), and hence admits two possible colorings. For each one, one can check whether such a coloring is a valid list-coloring in the associated connected component. This implies that the whole problem can be solved by a dynamic programming algorithm similar to the one given in \cite[Theorem 2]{refGravier02}: we consider the connected components one after the other, keeping track for each one of them of all the possible weights of vertices of color $c$ in the set $V_h$, for each $c$ and $h$, and then combining these weights whenever a new component is considered. The total amount of information that we need to keep track of is thus $O((\max_{h,c} W_{hc})^{2p})$, which is polynomial under our assumptions.

\subsection{When vertex weights are arbitrary}

We first give a short proof that {\sc WeightedLocallyBoundedListColoring} is \textbf{NP}-complete in general, even when $k=2$ (there are two colors) and $p=1$ (the vertex partition contains only one set). To do this, we consider the following well-known weakly \textbf{NP}-complete problem \cite{refGarey}:\\

\noindent {\sc Partition}\\
\textit{Instance}: $n+1$ positive integers $(a_1, \dots, a_{n}, B)$ such that $\sum_{i=1}^{n} a_i=2B$.\\
\textit{Question}: Decide whether there exists $I' \subset \{1, \dots, n\}$ such that $\sum_{i \in I'} a_i=B$.\\

Given an instance $(a_1, \dots, a_{n}, B)$ of {\sc Partition}, we define the following instance $I$ of {\sc WeightedLocallyBoundedListColoring}: the graph $G$ consists of $n$ isolated vertices $v_1, \dots, v_n$. We define $p=1$, $k=2$ and $W_{11}=W_{12}=B$. We also define, for each $i \in \{1, \dots, n\}$, $L(v_i)=\{1,2\}$ (i.e., any vertex can take any color) and $w(v_i)=a_i$. Then, $(a_1, \dots, a_{n}, B)$ has a solution if and only if $I$ admits a solution, i.e., a 2-coloring where the total weight of the vertices of each color is $B$: a vertex $v_i$ will take color 1 if and only if $i \in I'$. This implies:

\begin{Th}\label{th:NPC-partition}
{\sc WeightedLocallyBoundedListColoring} is \textbf{NP}-complete in the weak sense in graphs consisting of isolated vertices, even if $k=2$, $p=1$, and each vertex can take any color.
\end{Th}

Note that this instance can be made connected by setting $k=3$ and adding, for each $i \in \{1, \dots, n-1\}$, two edges $(v_i,u_i)$ and $(u_i,v_{i+1})$, where the $u_i$'s are new vertices of color 3 (e.g., $L(u_i)=\{3\}$ for each $i \in \{1, \dots, n-1\}$). The graph $G$ we obtain is then a tree (in fact, a chain), and so has tree-width 1.

Also note that, if $k=2$ and $G$ is connected, then solving {\sc WeightedLocallyBoundedListColoring} is a trivial task, since a connected bipartite graph has only two proper (vertex) colorings using two colors.

In \cite{refBodlaender05}, Bodlaender and Fomin managed to show that the equitable coloring problem can be solved in polynomial time in graphs of bounded tree-width (i.e., even if the number of colors is not fixed). In the next subsections, we rule out this possibility by proving that, even if all vertex weights are polynomially bounded, {\sc WeightedLocallyBoundedListColoring} remains \textbf{NP}-complete in graphs of bounded tree-width, if either $k$ or $p$ is not fixed.

\subsection{When the number of colors is not fixed}

First, we look at the case where the number $k$ of colors is not fixed. We consider the following well-known strongly \textbf{NP}-complete problem \cite{refGarey}:\\

\noindent {\sc $3-$Partition}\\
\textit{Instance}: A set of $3n$ positive integers $A=(a_1, \dots, a_{3n})$, and an integer $B$ such that $\sum_{i=1}^{3n} a_i=nB$ and, for each $i \in \{1, \dots, 3n\}$, $B/4 < a_i < B/2$.\\
\textit{Question}: Decide whether $A$ can be partitioned into $n$ disjoint sets of three elements $A_1, \dots, A_n$ such that, for each $j \in \{1, \dots, n\}$, $\sum_{a_i \in A_j} a_i = B$.\\

Since this problem is strongly \textbf{NP}-complete, we can assume w.l.o.g. that the $a_i$'s are polynomially bounded in $n$. We construct an instance of {\sc WeightedLocallyBoundedListColoring} in a graph $G$ consisting of $3n$ isolated vertices $v_i$, by defining  $p=1$, $k=n$, $W_{1c}=B$ for each $c \in \{1, \dots, n\}$, and, for each $i \in \{1, \dots, 3n\}$, $w(v_i)=a_i$ and $L(v_i)=\{1, \dots, k\}$. It is easy to see that the equivalence between the solutions of the two instances is given by: for each $h \in \{1, \dots, n\}$, a vertex $v_i$ will take color $h$ in $G$ if and only if $a_i \in A_h$. Hence:

\begin{Th}\label{th:NPC-3partition}
{\sc WeightedLocallyBoundedListColoring} is \textbf{NP}-complete in the strong sense in graphs consisting of isolated vertices, even if $p=1$, all the vertex weights are polynomially bounded, and each vertex can take any color.
\end{Th}

As in the case of Theorem \ref{th:NPC-partition}, we can make this instance connected (for instance, obtaining a chain, or a star) by defining $k=n+1$ and adding new vertices (and associated edges) that can only take color $n+1$. Note, however, that in this reduction we need non uniform vertex weights. Actually, this is unavoidable, from the following fact: the problem becomes tractable in graphs with no edges (and hence also in stars, by a simple reduction) when $w(v)=1$ for each $v \in V$. For the sake of completeness, we give a proof of this fact.

\begin{Proposition}\label{prop:isolated-vertices-poly}
If all vertex weights are 1, {\sc WeightedLocallyBoundedListColoring} is polynomial-time solvable in graphs consisting of isolated vertices.
\end{Proposition}
\begin{proof}
Given an instance in such a graph $G$, we can construct a bipartite graph $H$ as follows: there is a vertex $v_i$ in $H$ for each vertex $v_i$ in $G$, and, for each $h \in \{1, \dots, p\}$ and $c \in \{1, \dots, k\}$, there are $W_{hc}$ vertices $u^{j}_{h,c}$ in $H$. There is an edge between $v_i$ and $u^{j}_{h,c}$ in $H$ if and only if $v_i \in V_h$ and $c \in L(v_i)$: since $\sum_{c=1}^{k} W_{hc} = \sum_{v:v \in V_h} w(v) = \vert V_h \vert$ for each $h \in \{1, \dots, p\}$, this implies that the number of vertices is the same in each side of the bipartition of $H$.

Then, there exists a feasible coloring for the initial {\sc WeightedLocallyBoundedListColoring} instance in $G$ if and only if $H$ admits a perfect matching. ($H$ being bipartite, such a perfect matching can be computed as a maximum flow.) Indeed, we define the following equivalence: vertex $v_i \in V_h$ has color $c$ if and only if the edge $v_i u^{j}_{h,c}$ belongs to this perfect matching, for some $j$.

On the one hand, each vertex $v_i$ will be incident to one edge of the matching, and hence $v_i$ will take one color (by construction, this is a color in $L(v_i)$); on the other hand, each vertex $u^{j}_{h,c}$ will be incident to one edge of the matching, and hence in each $V_h$ there will be exactly $W_{hc}$ vertices with color $c$.
\end{proof}

We can nevertheless prove the strong \textbf{NP}-completeness of {\sc WeightedLocallyBoundedListColoring} in graphs of bounded tree-width when $w(v)=1$ for each $v \in V$, $p=1$ and $k$ is arbitrary, by using a more complex reduction from {\sc $3-$Partition}. Actually, such a reduction was already given in \cite[Theorem 8]{refBodlaender05} for the precolored variant of the equitable coloring problem. In this reduction, each instance is a set of trees of height 3 (together with additional isolated vertices). The leaves are in fact precolored vertices, which are used to restrain the set of possible colors for other vertices. Hence, if we consider list-colorings, we can ignore (i.e., remove) these leaves, and simply define suitable lists of possible colors for the other vertices. This yields instances consisting of trees of height at most 2, and leaves as open the case where each connected component is a tree of height at most 1 (which, unlike trees of height 2, are also cographs).

We now describe an alternative reduction, that needs more colors but only makes use of stars (i.e., trees of height 1), and that will prove useful in Sections \ref{sect:btw} to \ref{sect:split-graphs}. Assume we are given an instance $I=(a_1, \dots, a_{3n}, B)$ of {\sc $3-$Partition} with $n \geq 2$. We construct an instance $I'$ of {\sc WeightedLocallyBoundedListColoring} as follows: we define $p=1$, $k=3n^2+4n$, $w(v)=1$ for each vertex $v$, and the graph, that contains $3n^2(nB+1)+3n$ vertices, consists of $3n^2$ vertex-disjoint stars and $3n$ isolated vertices $u_1, \dots, u_{3n}$. We denote each star by $S_i^j$, for each $i \in \{1, \dots, 3n\}$ and each $j \in \{1, \dots, n\}$, and its central vertex by $v_i^j$. For each $i \in \{1, \dots, 3n\}$ and each $j \in \{1, \dots, n\}$, $v_i^j$ can take any of the $n+1$ colors in $\{n+i, ni+3n+1, ni+3n+2, \dots, ni + 4n\}$, and $S_i^j$ has $3na_i$ leaves, which can only take colors $j$ and $n+i$. Moreover, for each $i \in \{1, \dots, 3n\}$, vertex $u_i$ can take any of the $n$ colors in $\{ni+3n+1, \dots, ni+4n\}$. We also define $W_{1j}=3nB$ for each $j \in \{1, \dots, n\}$, $W_{1(n+i)}=3na_i(n-1)+1$ for each $i \in \{1, \dots, 3n\}$, and $W_{1(ni+3n+h)}=1$ for each $i \in \{1, \dots, 3n\}$ and each $h \in \{1, \dots, n\}$.

Then, the following holds:

\begin{Lm}
$I$ has a solution if and only if $I'$ has a solution.
\end{Lm}
\begin{proof}
We shall prove that the following equivalence holds: $a_i \in A_j$ if and only if $v_i^j$ takes color $n+i$, for each $i \in \{1, \dots, 3n\}$ and each $j \in \{1, \dots, n\}$.

Assume that we are given a solution to the {\sc $3-$Partition} instance $I$. For each $i \in \{1, \dots, 3n\}$, if $a_i \in A_j$ for some $j$, then $v_i^{j}$ takes color $n+i$, and the $3na_i$ leaves in $S_i^j$ take color $j$, while, for each $j' \neq j$, $v_i^{j'}$ takes color $ni+3n+j'$, and all the leaves in $S_i^{j'}$ take color $n+i$. Moreover, for each $i \in \{1, \dots, 3n\}$, $u_i$ takes color $ni+3n+j$. This implies, in $I'$, that there are $3na_i(n-1)+1$ vertices of color $n+i$ for each $i \in \{1, \dots, 3n\}$, and there are $3na_{i_1}+3na_{i_2}+3na_{i_3}=3nB$ (assuming $A_j = \{a_{i_1}, a_{i_2}, a_{i_3}\}$) vertices of color $j$ for each $j \in \{1, \dots, n\}$.

Conversely, assume that we are given a solution to the {\sc WeightedLocallyBoundedListColoring} instance $I'$, and let $a_i \in A_j$ if $v_i^j$ takes color $n+i$ for each $i \in \{1, \dots, 3n\}$ and each $j \in \{1, \dots, n\}$. Observe that, for each $i \in \{1, \dots, 3n\}$, exactly one of the $v_i^j$'s takes color $n+i$ (from the fact that $W_{1(ni+3n+h)}=1$ for each $h$, and from the possible colors that can be taken by $u_i$ and the $v_i^j$'s). This means that the leaves of such a $S_i^j$ \emph{must} take color $j$. For each $i \in \{1, \dots, 3n\}$, the leaves of all the other $S_i^j$'s must take color $n+i$, since $W_{1(n+i)}=3na_i(n-1)+1$. Hence, for each $j \in \{1, \dots, n\}$, there are exactly three $i$'s in $\{1, \dots, 3n\}$ such that $v_i^j$ has color $n+i$: if there are at most two such $i$, then $a_i < B/2$ for each $i$ implies that at most $3n(B-1)$ vertices have color $j$, a contradiction; if there are at least four such vertices, then $a_i > B/4$ for each $i$ implies that at least $3n(B+1)$ vertices have color $j$, a contradiction. The sum of these three $a_i$'s is then $W_{1j}/3n=B$, which concludes the proof.
\end{proof}

Hence, we have proved:
\begin{Th}\label{th:NPC-forests}
{\sc WeightedLocallyBoundedListColoring} is \textbf{NP}-complete in the strong sense in star forests, even if $p=1$ and $w(v)=1$ for each $v \in V$.
\end{Th}

\begin{Cor}\label{th:NPC-forests-cographs}
{\sc WeightedLocallyBoundedListColoring} is strongly \textbf{NP}-complete in cographs of tree-width 1, even if $p=1$ and $w(v)=1$  $\forall v \in V$.
\end{Cor}

Again, we can make the instance $I'$ connected by adding a new color, as well as one or more vertices that can only take this new color. After doing this, we can obtain a tree of height 2 (or a caterpillar), while recall that, from a direct consequence of Proposition \ref{prop:isolated-vertices-poly}, {\sc WeightedLocallyBoundedListColoring} is tractable in stars (i.e., in trees of height 1) when $w(v)=1$ for each $v \in V$.

Moreover, one of the features of the above reduction is that all the $u_i$'s and $v_i^j$'s must take different colors, which will be useful in Section \ref{sect:split-graphs}. However, one can use a simpler reduction, in which $\vert L(v) \vert = 2$ for each vertex $v$. Too see this, remove all the $u_i$'s, and set $k=4n$ and $L(v) = \{j,n+i\}$ for each vertex $v$ in each $S_i^j$. The color bounds are then $W_{1(n+i)}=3na_i(n-1)+1$ for each $i \in \{1, \dots, 3n\}$ and $W_{1j}=3nB+3(n-1)$ for each $j \in \{1, \dots, n\}$. As in the proof of Theorem \ref{th:NPC-forests}, and for similar reasons, we have $a_i \in A_j$ if and only if $v_i^j$ takes color $n+i$ for each $i \in \{1, \dots, 3n\}$ and each $j \in \{1, \dots, n\}$.

Note that the above reductions leave as open the case of chains (which are not cographs, except when they are also stars) where the vertex weights are 1, $p$ is fixed and $k$ is not. The reduction given in \cite[Theorem 3]{refGravier99} closes this gap, and shows that the problem is strongly \textbf{NP}-complete in this case, even when $p=1$ and $\vert L(v) \vert = 2$ for each $v \in V$. However, this reduction is rather complicated, in particular since all the chains do not play the same role. We now give a different kind of reduction to prove this result, in which we have $p=1$ and $\vert L(v) \vert = 2$ for each $v \in V$, as well as an additional restriction on the lists of possible colors:

\begin{Th}\label{th:NPC-linear-forests}
{\sc WeightedLocallyBoundedListColoring} is \textbf{NP}-complete in the strong sense in linear forests, even if $p=1$, $w(v)=1$ and $\vert L(v) \vert = 2$ for each $v \in V$, and all the vertices of each chain share one common possible color.
\end{Th}
\begin{proof}
Given a {\sc $3-$Partition} instance, instead of defining $3n^2$ stars as in Theorem \ref{th:NPC-forests}, we define $3n^2$ disjoint chains $\mu_i^j$, for each $i \in \{1, \dots, 3n\}$ and each $j \in \{1, \dots, n\}$. We also define $p=1$, $k=7n+3n^2$, and $w(v)=1$ for each $v \in V$. For each $i \in \{1, \dots, 3n\}$ and each $j \in \{1, \dots, n\}$, $\mu_i^j$ has $2a_i+2$ vertices, and ($i$) the first and the last vertices in $\mu_i^j$ can only take colors $(i-1)n+j+7n$ and $n+i$, ($ii$) the vertices with an odd index different from 1 in $\mu_i^j$ can only take colors $j$ and $n+i$, ($iii$) the vertices with an even index different from $2a_i+2$ in $\mu_i^j$ can only take colors $n+i$ and $4n+i$. Hence, for each vertex $v$ in $\mu_i^j$, we have $n+i \in L(v)$. To end the reduction, we define $W_{1(n+i)}=n(a_i+1)$ and $W_{1(4n+i)}=(n-1)a_i$ for each $i \in \{1, \dots, 3n\}$, $W_{1(7n+j+(i-1)n)}=1$ for each $i \in \{1, \dots, 3n\}$ and each $j \in \{1, \dots, n\}$, and $W_{1j}=B$ for each $j \in \{1, \dots, n\}$.

Note that, for each $i \in \{1, \dots, 3n\}$ and each $j \in \{1, \dots, n\}$, any feasible coloring has exactly one vertex (either the first one or the last one) in $\mu_i^j$ of color $7n+j+(i-1)n$, while there can be at most $a_i+1$ vertices of color $n+i$ in this chain. Since each $W_{1(n+i)}$ is equal to $n(a_i+1)$, this implies that in any feasible coloring there are exactly $a_i+1$ vertices of color $n+i$ in $\mu_i^j$, for each $i \in \{1, \dots, 3n\}$ and each $j \in \{1, \dots, n\}$: the vertices of color $n+i$ in $\mu_i^j$ are either all the vertices with an odd index, or all the vertices with an even index. Hence, in each chain, the $a_i$ remaining vertices either all have color $j$, or all have color $4n+i$. Moreover, since each $W_{1(4n+i)}$ is equal to $(n-1)a_i$, this implies that in any feasible coloring there are exactly $n-1$ integers $j$ in $\{1, \dots, n\}$ such that there are $a_i$ vertices of color $4n+i$ in $\mu_i^j$, for each $i \in \{1, \dots, 3n\}$.

We can thus define the following equivalence between the two instances: for each $i$ and each $j$, $a_i \in A_j$ if and only if there are $a_i$ vertices of color $j$ in $\mu_i^j$ (for each $j$, there will be exactly three such $i$'s, for reasons similar to the ones mentioned in the proof of Theorem \ref{th:NPC-forests}), which concludes the proof.
\end{proof}

We conclude by pointing out that we can turn the previous instance into a chain by adding a new color and new vertices that can only take this new color.

\subsection{When $p$ is not fixed}

Now, we prove the \textbf{NP}-completeness of the case where $k$ is fixed but $p$ is not, by considering the following well-known strongly \textbf{NP}-complete problem \cite{refGarey}:\\

\noindent {\sc MonotoneOneInThreeSAT}\\
\textit{Instance}: A set $X$ of $\nu$ boolean variables $x_1, \dots, x_{\nu}$, and a set of $\mu$ clauses, each one containing exactly three (non negated) boolean variables from $X$.\\
\textit{Question}: Decide whether there exists a truth assignment for the variables in $X$ such that in every clause there is exactly one variable equal to \emph{true}.\\

\begin{Th}\label{th:NPC-forest-1in3SAT}
{\sc WeightedLocallyBoundedListColoring} is \textbf{NP}-complete in the strong sense in star forests, even if $k=2$, and each vertex has weight 1 and can take any color.
\end{Th}
\begin{proof}
Given an instance of {\sc MonotoneOneInThreeSAT}, we construct an instance of {\sc WeightedLocallyBoundedListColoring} as follows: the graph $G$ consists of $\nu$ vertex-disjoint stars, one star per variable $x_i$. We denote by $v_i$ the central vertex of the $i$th star, and by $u^0_i, \dots, u^{occ(i)}_i$ its leaves, where $occ(i)$ is the number of occurrences of $x_i$ in the set of clauses. Then, we define $k=2$, $p=\nu+\mu$, and, for each vertex $v$, $w(v)=1$ and $L(v)=\{1, 2\}$. The partition of the vertices of $G$ and their target weights are given by: $V_h=\{v_h, u^0_h\}$ and $W_{h1}=W_{h2}=1$ for each $h \in \{1, \dots, \nu\}$, and $V_h=\{u^a_i, u^b_j, u^c_k\}$, $W_{h1}=1$ and $W_{h2}=2$ for each $h \in \{\nu+1, \dots, \nu+\mu\}$, where we consider that the $(h-\nu)$th clause consists of the $a$th occurrence of the variable $x_i$, of the $b$th occurrence of the variable $x_j$, and of the $c$th occurrence of the variable $x_k$.

It is easy to check that we have the following equivalence between the initial {\sc MonotoneOneInThreeSAT} instance and this {\sc WeightedLocallyBoundedListColoring} instance: for each $i \in \{1, \dots, \nu\}$, variable $x_i$ is equal to \emph{true} if and only if $v_i$ has color $2$, which concludes the proof.
\end{proof}

\begin{Cor}\label{th:NPC-forests-cographs-1in3SAT}
{\sc WeightedLocallyBoundedListColoring} is strongly \textbf{NP}-complete in cographs of tree-width 1, even if $k=2$, and each vertex has weight 1 and can take any color.
\end{Cor}

As in the case of Theorem \ref{th:NPC-forests}, we can make the previous instance connected (obtaining a tree of height 2, or a caterpillar) by defining $k=3$ and adding new vertices (and associated edges, to link the stars together) that can only take color 3 (for instance, we can define $p=\nu+\mu+1$ and $W_{p3} = \vert V_p \vert$, and include all these new vertices of weight 1 in the new $V_p$, keeping the other $V_h$'s unchanged).

Again, this reduction, that will prove useful both in Sections \ref{sect:btw} and \ref{sect:cographs} (thanks to Corollary \ref{th:NPC-forests-cographs-1in3SAT}), leaves as open the case of chains where $w(v)=1$ for each $v \in V$, $k$ is fixed and $p$ is not. The next theorem closes this gap.

\begin{Th}\label{th:NPC-linear-forest-1in3SAT}
{\sc WeightedLocallyBoundedListColoring} is \textbf{NP}-complete in the strong sense in linear forests, even if $k=2$, and each vertex has weight 1 and can take any color.
\end{Th}
\begin{proof}
Given a {\sc MonotoneOneInThreeSAT} instance, we define a {\sc WeightedLocallyBoundedListColoring} instance as follows: the graph $G$ consists of $\nu$ vertex-disjoint chains (one per variable $x_i$) and of isolated vertices. The vertices of the $i$th chain are $v_i^1, u_i^1, \dots, v_i^{occ(i)}, u_i^{occ(i)}$ (in this order), where $occ(i)$ is defined as in the proof of Theorem \ref{th:NPC-forest-1in3SAT}. For each chain, we also add $occ(i)$ isolated vertices ${v'}_{i}^{1}, \dots, {v'}_{i}^{occ(i)}$. Then, we define $k=2$, $p=\mu+\sum_{i=1}^{\nu} occ(i)$, and, for each vertex $v$, $w(v)=1$ and $L(v)=\{1, 2\}$. It remains to define the partition of the vertices of $G$ and their target weights. We have $V_h=\{u^a_i, u^b_j, u^c_k\}$, $W_{h1}=1$ and $W_{h2}=2$ for each $h \in \{1, \dots, \mu\}$, where we consider that the $h$th clause consists of the $a$th occurrence of the variable $x_i$, of the $b$th occurrence of the variable $x_j$, and of the $c$th occurrence of the variable $x_k$. For each $i \in \{1, \dots, \nu\}$ and each $l \in \{1, \dots, occ(i)\}$, we define $V_h=\{v_i^l, {v'}_i^l\}$ and $W_{h1}=W_{h2}=1$, where $h = l + \mu + \sum_{j=1}^{i-1} occ(j)$.

As in the proof of Theorem \ref{th:NPC-forest-1in3SAT}, for each $i \in \{1, \dots, \nu\}$, variable $x_i$ is equal to \emph{true} if and only if $v_i^1$ has color $2$. Indeed, for each $i$, all the $v_i^j$'s must have the same color, and all the $u_i^j$'s must have the same color (the ${v'}_i^j$'s are only useful to ensure that $v_i^1$ can freely take color 1 or 2). Theorem \ref{th:NPC-linear-forest-1in3SAT} follows.
\end{proof}

As in the case of Theorem \ref{th:NPC-linear-forests}, we can make the previous instance connected (obtaining a chain) by defining $k=3$ and adding new vertices (and associated edges) that can only take color 3. Also note that, in the reductions used to prove Theorems \ref{th:NPC-forest-1in3SAT} and \ref{th:NPC-linear-forest-1in3SAT}, we only need colorings (not list-colorings), although we did need to make use of list-colorings in the proofs of Theorems \ref{th:NPC-forests} and \ref{th:NPC-linear-forests}.

\section{An algorithm for graphs of bounded tree-width}\label{sect:btw}

This section deals with {\sc WeightedLocallyBoundedListColoring} in partial $k$-trees. Given a graph $G=(V,E)$, a \emph{tree decomposition} of $G$ is a pair $\big(\{X_i | i \in I\},T\big)$ where $X_i \subseteq V, \forall i \in I$, is a \emph{bag} of vertices of $G$, $T=(I,F)$ is a tree, and:
\begin{enumerate}
 \item[(1)] $\bigcup_{i \in I} X_i = V$,
 \item[(2)] For every edge $uv \in E$, there is an $i \in I$ such that $u,v \in X_i$,
 \item[(3)] For all $i,j,l \in I$, if $j$ lies in the path between $i$ and $l$, then $X_i \cap X_l \subseteq X_j$.
\end{enumerate}

The \emph{width} of a given tree decomposition of a graph $G$ is equal to $\max_{i \in I}\vert X_i \vert - 1$. The \emph{tree-width} of a graph $G$, denoted by $tw(G)$, is the minimum width of a tree decomposition of $G$, taken over all tree decompositions of $G$. Note that trees (and hence chains and stars) have tree-width 1. Without loss of generality, we can also assume that the tree decomposition is \emph{nice} \cite{refKloks94}, i.e.:
\begin{itemize}
 \item $T$ is rooted at some node $r$,
 \item $T$ is binary and has $O(\vert V \vert)$ nodes,
 \item If a node $i$ has two children $j$ and $k$ then $X_i=X_j=X_k$\hfill (join node)
 \item If a node $i$ has one child $j$, then either
 \begin{enumerate}
  \item[(a)] $\vert X_i\vert = \vert X_j\vert -1$ and $X_i \subset X_j$\hfill (forget node)
  \item[(b)] $\vert X_i\vert = \vert X_j\vert +1$ and $X_j \subset X_i$\hfill (introduce node)
 \end{enumerate}
\end{itemize}
Given two vertices $i,j$ of $T$, we will use the notation $j\succeq i$ to denote the fact that $j$ is either $i$ or a descendant of $i$ with respect to $r$. Given a node $i \in I$, let $Y_i = E \cap (X_i \times X_i)$, i.e., $Y_i$ is the subset of $E$ induced by the vertices in $X_i$. Moreover, let $T_i = \bigcup_{j \succeq i} X_j$ and let $G[T_i]$ be the subgraph of $G$ induced by $T_i$.\\

In order to design a standard dynamic programming algorithm that solves {\sc WeightedLocallyBoundedListColoring}, we use the following function $f$:
\begin{itemize}
\item $f(i, c_i, \omega_{11}, \dots, \omega_{1k}, \dots, \omega_{p1}, \dots, \omega_{pk})$ = \emph{true} if there exists a list-coloring of $G[T_i]$ where each vertex $u \in X_i$ has color $c_i(u)$ and where the total weight of vertices of $G[T_i]$ having color $c$ in $V_h$ is $\omega_{hc}$, and \emph{false} otherwise.
\end{itemize}

To describe our algorithm, we now simply need to write down the induction equations defining the values of $f(\cdot)$, for each type of nodes of $T$. Then, these values will be computed in a bottom-up fashion, starting from the leaves of $T$.

Note that, by the definition of tree decompositions, any two vertices linked by an edge in $G$ must both belong to at least one common bag of the tree decomposition (Condition (2)). This implies that a (list-)coloring that is proper in each bag is also proper in the whole graph $G$, provided that each vertex has the same color in each bag it belongs to. Moreover, Condition (3) ensures that the subgraph of $T$ induced by the bags any given vertex belongs to is connected, and hence in all these bags this vertex will have the same color if we do not change its color whenever we move (in $T$) from one bag to an adjacent one.

\paragraph*{If $i$ is a forget node.}
Let $j$ denote its child such that $X_j=X_i \cup \{v\}$:
\begin{equation*}
f(i, c_i, \omega_{11}, \dots, \omega_{pk})=\bigvee_{c_j:(c_j(u) = c_i(u)~\forall u \in X_i) \wedge (c_j(v)\in L(v))} f(j, c_j, \omega_{11}, \dots, \omega_{pk})
\end{equation*}

\paragraph*{If $i$ is an introduce node.}
Let $j$ denote its child such that $X_j=X_i \setminus \{v\}$, and assume that $v \in V_h$:
\begin{equation*}
f(i, c_i, \omega_{11}, \dots, \omega_{pk}) = \mathcal{E} \bigwedge f(j, c_j, \omega_{11}, \dots, \omega_{hc_i(v)}-w(v), \dots, \omega_{pk})\\
\end{equation*}
where $\mathcal{E} = (c_i(v) \in L(v)) \bigwedge (c_i(v) \neq c_i(u)~\forall u:uv \in Y_i) \bigwedge (c_j(u)=c_i(u)~\forall u \in X_j) \bigwedge (\omega_{hc_i(v)} \geq w(v))$.\\

\paragraph*{If $i$ is a join node.}
Let $j$ and $l$ denote its two children, and let $w^i_{hc}$ be the total weight of vertices $v \in X_i \cap V_h$ such that $c_i(v)=c$, for each $h$ and $c$.
\begin{equation*}
  f(i, c_i, \omega_{11}, \dots, \omega_{pk})=\bigvee_{(q_{11},\dots,q_{pk}) \in Q^i}\big(f(j, c_i, \omega_{11}+w^i_{11}-q_{11}, \dots, \omega_{pk}+w^i_{pk}-q_{pk}) \wedge f(l, c_i, q_{11}, \dots, q_{pk})\big)
\end{equation*}
where $Q^i=\{(q_{11},\dots,q_{pk}):w^i_{hc} \leq q_{hc}\leq \omega_{hc}~\forall h~\forall c\}$.\\

\paragraph*{If $i$ is a leaf node.}
In this case, we just have to check that the coloring function $c_i$ provides a valid locally bounded list-coloring of $G[X_i]$.
\begin{equation*}
f(i, c_i, \omega_{11}, \dots, \omega_{pk})= \mathcal{E} \bigwedge (w(\{v\in X_i\cap V_h:c_i(v)=c\}) = \omega_{hc}~\forall h~\forall c)
\end{equation*}
where $\mathcal{E}=(c_i(v) \in L(v)~\forall v \in X_i) \bigwedge (c_i(v) \neq c_i(u)~\forall u,v \in X_i:uv \in Y_i)$.

\paragraph*{Root value.} The answer \emph{true} or \emph{false} for the initial instance is obtained at the root $r$ by computing the following value:
\begin{equation*}
\bigvee_{c_r:c_r(u)\in L(u)~\forall u \in X_r} f(r, c_r, W_{11}, \dots, W_{pk})
\end{equation*}

\begin{Lm}\label{lem:algo-btw}
The values of $f(i, c_i, \omega_{11}, \dots, \omega_{1k}, \dots, \omega_{p1}, \dots, \omega_{pk})$ computed by the above algorithm are correct.
\end{Lm}
\begin{proof}
In order to show that a vertex list-coloring is proper in $T_i$, first notice that, from our preliminary remark, it suffices to show that ($i$) it is a proper list-coloring in $X_i$ and ($ii$) the color of any vertex remains the same when moving from one bag of $T_i$ to an adjacent one. We now show the correctness of the above equations by considering each possible node type for $i$.

Assume $i$ is a forget node. Then, since $T_i = T_j$, $f(i, c_i, \omega_{11}, \dots, \omega_{pk})$ is true if and only if $f(j, c_j, \omega_{11}, \dots, \omega_{pk})$ is true for some coloring $c_j$, such that $c_j(u) = c_i(u)$ for each $u \in X_i$ (each vertex must keep its color when moving from $X_j$ to $X_i$) and vertex $v \in X_j$ takes some color in $L(v)$.

Assume $i$ is an introduce node. Then, $f(i, c_i, \omega_{11}, \dots, \omega_{pk})$ is true if and only if $f(j, c_j, \omega_{11}, \dots, \omega_{hc_i(v)}-w(v), \dots, \omega_{pk})$ is true ($v \in V_h$ has weight $w(v)$ and color $c_i(v)$, so the total weight of vertices of color $c_i(v)$ in $V_h$ is $\omega_{hc_i(v)}-w(v) \geq 0$ in $T_j$), each vertex keeps its color when moving from $X_j$ to $X_i$, and the color of $v$ defines a valid list-coloring in the subgraph of $G$ induced by $X_i$ (i.e., $c_i(v) \in L(v)$ and $c_i(v) \neq c_i(u)$ for each $u$ such that $uv \in Y_i$).

Assume $i$ is a join node. Then, $f(i, c_i, \omega_{11}, \dots, \omega_{pk})$ is true if and only if both $f(j, c_i, q'_{11}, \dots, q'_{pk})$ and $f(l, c_i, q_{11}, \dots, q_{pk})$ are true for some $q_{11}, \dots, q_{pk}$ and $q'_{11}, \dots, q'_{pk}$ such that $q_{hc}+q'_{hc}=\omega_{hc}+w^i_{hc}$ for each $h$ and $c$ (with obviously $w^i_{hc} \leq q_{hc} \leq \omega_{hc}$ and $w^i_{hc} \leq q'_{hc} \leq \omega_{hc}$ for each $h$ and $c$), since the weights of the vertices in $X_i$ are counted twice, i.e., both in $T_j$ and in $T_l$ (and any other vertex weight in $T_j \cup T_l$ is counted only once).

Assume $i$ is a leaf node. Then, $f(i, c_i, \omega_{11}, \dots, \omega_{pk})$ is true if and only if the coloring function $c_i$ provides a valid locally bounded list-coloring of $G[X_i]$.

Finally, the root value is obtained by requiring that $c_r$ is a list-coloring in the subgraph of $G$ induced by $X_r$, which concludes the proof.
\end{proof}

\paragraph*{Running time.}
Let $W_{\max}=\max_{h,c} W_{hc}$ (with $W_{\max} \leq n \max_{v \in V} w(v)$). The running time for a given node of $T$, that depends on its type, is given by:
\\
\begin{center}
 \begin{tabular}{|l|l|}
\hline
 node type & running time\\
\hline
forget & $O(tw(G)+k)$\\
introduce & $O(tw(G)+k)$\\
join & $O(W_{\max}^{pk})$ \\
leaf &$O(tw(G)(tw(G)+pk))$\\
\hline
\end{tabular}
\end{center}
~\\
~\\
There are $O(n)$ nodes in $T$, $O(k ^{tw(G)+1})$ possible colorings of any given bag, and $O(W_{\max}^{pk})$ possible $pk$-tuples $\omega_{11}, \dots, \omega_{pk}$, so when running the algorithm we have $O(n W_{\max}^{pk} k ^{tw(G)+1})$ values $f(\cdot)$ to compute. Since computing the optimal value only takes $O(k ^{tw(G)+1})$ time, the overall running time is $O(n W_{\max}^{pk} k ^{tw(G)+1}\\(W_{\max}^{pk}+tw(G)(tw(G)+pk)))$. Together with Lemma \ref{lem:algo-btw}, this implies:

\begin{Th}\label{th:algo-btw}
In graphs of bounded tree-width, {\sc WeightedLocallyBoundedListColoring} can be solved:
\begin{itemize}
\item in pseudopolynomial time when $p$ and $k$ are fixed,
\item in polynomial time when (i) $p$ and $k$ are fixed, and (ii) all vertex weights are polynomially bounded.
\end{itemize}
\end{Th}

Observe that these results are best possible, in the sense that, from Theorems \ref{th:NPC-partition} to \ref{th:NPC-linear-forest-1in3SAT}, dropping any of the assumptions (on $p$, $k$, and the vertex weights) leads to \textbf{NP}-completeness. They also generalize the results in \cite{refBentz09,refGravier02,refJansen}.

We close this section by mentioning that this approach can be adapted to solve an optimization version of {\sc WeightedLocallyBoundedListColoring}.

More precisely, one can associate a profit function $\pi: V \times \{1, \dots, k\} \rightarrow \mathbf Z$ to the vertices of $G$; the profit of a vertex then depends on its color. By slightly modifying the above dynamic programming algorithm, one can compute a valid weighted and locally bounded list-coloring of maximum (or minimum) profit (if any). In this case, the value of $f(\cdot)$ is no longer equal to \emph{true} or \emph{false}, i.e.:
\begin{itemize}
\item $f(i, c_i, \omega_{11}, \dots, \omega_{1k}, \dots, \omega_{p1}, \dots, \omega_{pk})$ = the maximum total profit of a list-coloring of $G[T_i]$ where each vertex $u \in X_i$ has color $c_i(u)$ and where the total weight of vertices of $G[T_i]$ having color $c$ in $V_h$ is $\omega_{hc}$ (if any).
\end{itemize}

In order to compute the value of this ``new'' function $f(\cdot)$, we must make some changes in the equations. We provide them without proofs, as the arguments are quite similar to the ones used in the proof of Lemma \ref{lem:algo-btw} (note that, by convention, infeasible solutions will have a value of $- \infty$, as we maximize $f(\cdot)$):
\begin{itemize}
\item (in forget nodes and root value) ``$\bigvee$'' becomes ``$\max$''.
\item (in introduce nodes) $f(i, c_i, \omega_{11}, \dots, \omega_{1k}, \dots, \omega_{pk})$ is equal to $- \infty$ if $\mathcal{E}$ is \emph{false}, and to $f(j, c_j, \omega_{11}, \dots, \omega_{hc_i(v)}-w(v), \dots, \omega_{pk})+\pi(v,c_i(v))$ otherwise.
\item (in join nodes) ``$\bigvee$'' becomes ``$\max$'', ``$\wedge$'' becomes ``$+$'', and we add the value ``$-\sum_{v \in X_i} \pi(v,c_i(v))$'' at the end of the line.
\item (in leaf nodes) $f(i, c_i, \omega_{11}, \dots, \omega_{1k}, \dots, \omega_{pk})$ is equal to $\sum_{v \in X_i} \pi(v,c_i(v))$ if both $\mathcal{E}$ and the condition ``$w(\{v\in X_i\cap V_h:c_i(v)=c\}) = \omega_{hc}~\forall h~\forall c$'' are \emph{true}, and to $- \infty$ otherwise.
\end{itemize}

\section{Locally bounded list-colorings in cographs}\label{sect:cographs}

In this section, we study the tractability of {\sc WeightedLocallyBoundedListColoring} in cographs. Cographs, as defined in Section \ref{sec:Intro}, can be characterized in several ways. For instance, a graph $G=(V,E)$ is a cograph if and only if it can be associated with a \emph{cotree} $T$: the leaves of $T$ are the vertices of $G$, and the internal nodes of $T$ are either \emph{union nodes} or \emph{join nodes}. A subtree of $T$ having a \emph{union node} as a root corresponds to the disjoint union of the subgraphs of $G$ associated with the children of this node, and a subtree of $T$ having a \emph{join node} as a root corresponds to the complete union of the subgraphs of $G$ associated with the children of this node (i.e., we add an edge between every pair of vertices with one vertex in each subgraph). Moreover, this cotree can easily be transformed in linear time into a binary cotree with $O(\vert V \vert)$ nodes \cite{refBrandstadt}.

First note that {\sc WeightedLocallyBoundedListColoring} is still \textbf{NP}-complete in cographs, even when $k$ is arbitrary, $p=1$, and each vertex has weight 1. Indeed, on the one hand, it was proved in \cite{refJansen} that the list-coloring problem is \textbf{NP}-complete in complete bipartite graphs (which are cographs), when $k$ is not fixed. On the other hand, the bounded coloring problem was proved to be \textbf{NP}-complete in cographs in \cite{refBodlaenderJansen95}, by a reduction from bin packing (which, by \cite{refJansenMarx13}, also shows its \textbf{W[1]}-hardness with respect to $k$ in cographs).

However, the instances used in \cite{refBodlaenderJansen95,refJansen} have a large tree-width (as otherwise the list-coloring and bounded coloring problems are tractable \cite{refBodlaender05,refJansen}): since star forests and isolated vertices are cographs, Corollary \ref{th:NPC-forests-cographs} shows that this remains true even when the tree-width is 1 (the case of tree-width 0 and polynomially bounded vertex weights being covered by Theorem \ref{th:NPC-3partition}). Corollary \ref{th:NPC-forests-cographs-1in3SAT} shows that, under the same assumption of being a cograph of tree-width 1 (the case of tree-width 0 will be discussed later in this section), this is also true as soon as $k=2$ (provided that $p$ is arbitrary), even when each vertex has weight 1 and can take any color. Finally, as in the case of partial $k$-trees, Theorem \ref{th:NPC-partition} shows that allowing arbitrary vertex weights leads to weak \textbf{NP}-completeness in cographs of tree-width 0, even if $k=2$, $p=1$, and any vertex can take any color.

Moreover, the instances from the previous reductions can be made connected by adding a new vertex, adjacent to all the other vertices, that must take a new color (this increases the tree-width by 1). In particular, when the graph in the reduction consisted of isolated vertices, it then becomes a star.

However, when both $p$ and $k$ are fixed, we can design an efficient dynamic programming algorithm, based on standard techniques, to solve {\sc WeightedLocallyBoundedListColoring} in cographs, by using the associated (binary) cotrees. In order to describe this algorithm, we define the following function:
\begin{itemize}
\item $f'(i, \omega_{11}, \dots, \omega_{1k}, \dots, \omega_{p1}, \dots, \omega_{pk})$ = \emph{true} if there exists a list-coloring of the subgraph of $G$ induced by the leaves of the subgraph of $T$ rooted at node $i$, where the total weight of vertices of this induced subgraph of $G$ having color $c$ in $V_h$ is $\omega_{hc}$, and \emph{false} otherwise.
\end{itemize}

The value of each $f'(\cdot)$ is then computed in a bottom-up fashion, as follows.

\paragraph*{If $i$ is a join node.}
Let $j$ and $l$ denote its two children.
\begin{equation*}
  f'(i, \omega_{11}, \dots, \omega_{pk})=\bigvee_{(q_{11},\dots,q_{pk}) \in Q_{join}^i}\big(f'(j, \omega_{11}-q_{11}, \dots, \omega_{pk}-q_{pk}) \wedge f'(l, q_{11}, \dots, q_{pk})\big)
\end{equation*}
where $Q_{join}^i=\{(q_{11},\dots,q_{pk}):0 \leq q_{hc}\leq \omega_{hc}~\forall h~\forall c\} \cap \{(q_{11},\dots,q_{pk}): (\sum_{h=1}^p q_{hc}) \times (\sum_{h=1}^p (\omega_{hc}-q_{hc}))=0 ~\forall c\}$.

\paragraph*{If $i$ is a union node.}
Let $j$ and $l$ denote its two children.
\begin{equation*}
  f'(i, \omega_{11}, \dots, \omega_{pk})=\bigvee_{(q_{11},\dots,q_{pk}) \in Q_{union}^i}\big(f'(j, \omega_{11}-q_{11}, \dots, \omega_{pk}-q_{pk}) \wedge f'(l, q_{11}, \dots, q_{pk})\big)
\end{equation*}
where $Q_{union}^i=\{(q_{11},\dots,q_{pk}):0 \leq q_{hc}\leq \omega_{hc}~\forall h~\forall c\}$.

\paragraph*{If $i$ is a leaf node.}
Let $v_i \in V_h$ be the vertex of $G$ corresponding to this leaf.
\begin{equation*}
  f'(i, \omega_{11}, \dots, \omega_{pk})=\bigvee_{j \in L(v_i)} \big(\omega_{hj}=w(v_i) \wedge (\omega_{h'j'}=0~\forall (h',j') \neq (h,j))\big)
\end{equation*}

\begin{Lm}\label{lem:algo-cographs}
The values of $f'(i, \omega_{11}, \dots, \omega_{1k}, \dots, \omega_{p1}, \dots, \omega_{pk})$ computed by the above algorithm are correct.
\end{Lm}
\begin{proof}
In order to show that a vertex list-coloring is proper in the whole graph, it is sufficient to prove that it is a valid list-coloring at each step (i.e., during the computation of the value of $f'(\cdot)$ at each node of the cotree).

Assume $i$ is a join node. Then, the graph induced by the leaves of the subgraph of $T$ rooted at node $i$ is the complete union of two graphs. For each $h$ and $j$, the sum of the weights of the vertices of color $j$ in $V_h$ in these two graphs must be equal to $\omega_{hc}$. However, when taking the complete union of two subgraphs of $G$, vertices not belonging to the same subgraph cannot have the same color (otherwise, the coloring would not be proper).

Assume $i$ is a union node. Then, the graph induced by the leaves of the subgraph of $T$ rooted at node $i$ is the disjoint union of two graphs, and thus, for each $h$ and $j$, the sum of the weights of the vertices of color $j$ in $V_h$ in these two graphs must be equal to $\omega_{hc}$.

Assume $i$ is a leaf node. Then, in any valid list-coloring, the vertex $v_i$ of $G$ associated with $i$ takes only one color, which belongs to $L(v_i)$.
\end{proof}

\noindent The above algorithm runs in $O(n W_{\max}^{pk} (W_{\max}^{pk}+pk))$ time, where $W_{\max}=\max_{h,c} W_{hc}$. Together with Lemma \ref{lem:algo-cographs}, this yields the following result:

\begin{Th}\label{th:algo-cographs}
{\sc WeightedLocallyBoundedListColoring} can be solved:
\begin{itemize}
\item in pseudopolynomial time in cographs when $p$ and $k$ are fixed,
\item in polynomial time in cographs when (i) $p$ and $k$ are fixed, and (ii) all vertex weights are polynomially bounded.
\end{itemize}
\end{Th}

Note that this result generalizes the ones in \cite{refBodlaenderJansen95,refGravier02,refJansen}. Again, one can associate a profit function $\pi$ to the vertices of $G$, and modify slightly this dynamic programming algorithm (by replacing ``$\bigvee$'' by ``$\max$'' and ``$\wedge$'' by ``$+$'' in union and join nodes, and by returning $\max_{j \in L(v_i)} \pi(v_i, j)$ in leaf nodes if the above conditions are satisfied, and $- \infty$ otherwise) in order to compute a feasible weighted and locally bounded list-coloring of maximum profit (if any). Note that we return $- \infty$ if $Q_{union}^i=\emptyset$ or $Q_{join}^i=\emptyset$ for some $i$.\\

Theorem \ref{th:algo-cographs} can also be used to show the following proposition:

\begin{Proposition}\label{prop:isolated-vertices-k-fixed}
If $k=O(1)$, {\sc WeightedLocallyBoundedListColoring} can be solved in (pseudo)polynomial time in graphs consisting of isolated vertices.
\end{Proposition}
\begin{proof}
Since there are no edges in this case, any coloring will be proper, and hence any set $V_i$ of the partition can be considered independently from the $p-1$ other $V_j$'s. Hence, solving such an instance of {\sc WeightedLocallyBoundedListColoring} is then equivalent to solving $p$ independent instances where the partition contains only one vertex set. Since $k$ is fixed as well, any such instance can be solved in (pseudo)polynomial time thanks to Theorem \ref{th:algo-cographs}.
\end{proof}

To close this section, we study {\sc WeightedLocallyBoundedListColoring} in particular cographs. We begin by studying complete bipartite graphs $K_{n_1,n_2}$ (which are cographs represented by binary cotrees containing $n_1+n_2-2$ union nodes and only one join node, at the root), where this problem is trivial when $k=2$ (so we shall assume $k \geq 3$). Since the complexity of {\sc WeightedLocallyBoundedListColoring} is the same in graphs consisting of isolated vertices and in stars (which are the graphs $K_{1,n_2}$), Theorems \ref{th:NPC-partition} and \ref{th:NPC-3partition} imply that {\sc WeightedLocallyBoundedListColoring} is \textbf{NP}-complete in complete bipartite graphs if $k$ is not fixed or if the vertex weights are not polynomially bounded, even if any vertex can take any color. Moreover, it was proved in \cite{refJansen} that the list-coloring problem (without constraints on the sizes of the color classes) is strongly \textbf{NP}-complete in complete bipartite graphs when the number $k$ of colors is arbitrary. However, adding isolated vertices cannot yield complete bipartite graphs, so this result does not directly apply to our problem. We now show that a more complex reduction, partly inspired by the one given in \cite{refJansen}, can be obtained for {\sc WeightedLocallyBoundedListColoring} in this case:

\begin{Th}\label{th:NPC-complete-bipartite}
{\sc WeightedLocallyBoundedListColoring} is strongly \textbf{NP}-complete in complete bipartite graphs, even if $p$ and all vertex weights are 1.
\end{Th}
\begin{proof}
Given a {\sc MonotoneOneInThreeSAT} instance $I$ with $\mu$ clauses and $\nu$ variables $x_1, \dots, x_{\nu}$, we construct a {\sc WeightedLocallyBoundedListColoring} instance $I'$ in a complete bipartite graph $G$ as follows.

On the left side of $G$, there are $\mu$ vertices $u_1, \dots, u_{\mu}$ associated with the $\mu$ clauses, and, for each $i \in \{1, \dots, \nu\}$, there are $occ(i)$ vertices $v_i^1, \dots, v_i^{occ(i)}$ associated with each variable $x_i$, where $occ(i)$ is the number of occurrences of $x_i$ in the set of clauses. On the right side of $G$, there are $occ(i)$ vertices $w_i^1, \dots, w_i^{occ(i)}$ associated with each variable $x_i$. We set $p=1$, $k=2\nu+1$, $w(v)=1$ for each vertex $v$, and, for each $i \in \{1, \dots, \nu\}$ and each $j \in \{1, \dots, occ(i)\}$, we define $L(v_i^j)=\{i,2\nu+1\}$ and $L(w_i^j)=\{i,\nu+i\}$. Furthermore, for each $h \in \{1, \dots, \mu\}$, we set $L(u_h)=\{\nu+i_1, \nu+i_2, \nu+i_3\}$, assuming that the $h$th clause is $x_{i_1} \vee x_{i_2} \vee x_{i_3}$. Finally, we set $W_{1i}=W_{1(\nu+i)}=occ(i)$ for each $i \in \{1, \dots, \nu\}$, and $W_{1(2\nu+1)}=\mu$. Then, we have the following equivalence between the solutions of $I$ and $I'$: for each $i \in \{1, \dots, \nu\}$, all the $w_i^j$'s take color $i$ if and only if variable $x_i$ takes value \emph{true}. Let us justify such an equivalence.

First assume that we know a solution to $I$. For each $h \in \{1, \dots, \mu\}$, if we denote by $x_{i_1} \vee x_{i_2} \vee x_{i_3}$ the $h$th clause, then vertex $u_h$ takes color $\nu+i_l$, where $i_l \in \{i_1,i_2,i_3\}$ is such that $x_{i_l}$ is the only variable of value \emph{true} in the $h$th clause. Moreover, for each $i \in \{1, \dots, \nu\}$, all the $v_i^j$'s take color $2\nu+1$ if $x_i$ has value \emph{true}, and color $i$ otherwise (meaning that the number of vertices of color $2\nu+1$ is exactly the number of occurrences of variables of value \emph{true}). It can be easily checked that this yields a feasible solution to $I'$, since there are exactly $\mu$ vertices of color $2\nu+1$ (as each of the $\mu$ clauses has exactly one true literal).

Assume now that there is a solution to $I'$. Consider any $i \in \{1, \dots, \nu\}$. If there is some $w_i^j$ that takes color $\nu+i$, then, by construction, no vertex $u_h$ can take color $\nu+i$. Hence, the only vertices that can take color $\nu+i$ are the $w_i^j$'s. So, the only way to have $W_{1(\nu+i)}=occ(i)$ vertices of color $\nu+i$ is that \emph{each} vertex $w_i^j$ takes color $\nu+i$. This implies, in turn, that the only way to have $W_{1i}=occ(i)$ vertices of color $i$ is that \emph{each} vertex $v_i^j$ takes color $i$. Now, if there is some $w_i^j$ that takes color $i$, then, by construction, \emph{each} vertex $v_i^j$ \emph{must} take color $2 \nu + 1$. So, the only way to have $W_{1i}=occ(i)$ vertices of color $i$ is that \emph{each} vertex $w_i^j$ takes color $i$. This implies, in turn, that the only way to have $W_{1(\nu+i)}=occ(i)$ vertices of color $\nu+i$ is that $occ(i)$ vertices $u_h$ take color $\nu+i$. The only such vertices that can take color $\nu+i$ are the $occ(i)$ ones associated with the clauses where $x_i$ appears, so \emph{all} these $u_h$'s must take color $\nu+i$.

In short, for each $i$, either all the $w_i^j$'s take color $i$ (in which case the $occ(i)$ vertices $v_i^j$ take color $2\nu+1$, and the $occ(i)$ vertices $u_h$ corresponding to clauses where $x_i$ appears take color $\nu+i$), or all the $w_i^j$'s take color $\nu+i$ (in which case no vertex $u_h$ takes color $\nu+i$, and the $occ(i)$ vertices $v_i^j$ take color $i$). Hence, whenever a vertex $u_h$ (associated with the clause $x_{i_1} \vee x_{i_2} \vee x_{i_3}$) takes color $\nu+i_l$ for some $i_l \in \{i_1,i_2,i_3\}$, then this means that all the $w_{i_l}^j$'s take color $i_l$, and that all the $w_{i_{l'}}^j$'s take color $\nu+i_{l'}$ for each $i_{l'} \in \{i_1,i_2,i_3\} \setminus \{i_l\}$. In other words, there is exactly one literal of value \emph{true} in each clause of $I$.
\end{proof}

On the positive side, we show that {\sc WeightedLocallyBoundedListColoring} can be solved in (pseudo)polynomial time in complete bipartite graphs if $k$ is fixed but $p$ is not. Note that, in such a graph, all the vertices of a given color belong to the same side of the bipartition of $G$. So, if $k$ is fixed, we can ``guess'' which colors are on each side by enumerating all the possibilities (there are $2^k$ such possibilities). For each configuration (i.e., for each such assignment of colors to both sides) we enumerate, we consider each side of the bipartition independently, and for each one we check whether this configuration is feasible (note that, for each side, we can consider each $V_i$ independently, since they do not interact). We are left with a set of independent instances where $p=1$, $k$ is fixed and all the vertices are isolated vertices, so we can solve each one of them in (pseudo)polynomial time thanks to Proposition \ref{prop:isolated-vertices-k-fixed}. This yields:

\begin{Th}\label{th:complete-bipartite-graphs-k-fixe}
If $k=O(1)$, {\sc WeightedLocallyBoundedListColoring} can be solved in (pseudo)polynomial time in complete bipartite graphs.
\end{Th}

Now, we turn to the case of complete graphs $G=K_n$ (which are cographs represented by binary cotrees containing $n-1$ join nodes and no union node). Notice that, in a complete graph, all the nodes must take different colors. So, we must have $k \geq n$. If $k > n$, then in any feasible coloring one of the colors will not be used: this is possible only if we have $\sum_{h=1}^p W_{hc}=0$ for this color $c$, and thus we can remove it. Hence, we know that we necessarily have $k=n$, and that each color will appear exactly once. In particular, since the $V_i$'s are disjoint, this implies that, for each color $c$, we have $W_{h_c c} > 0$ for some $h_c$ and $W_{h'c}=0$ for each $h' \neq h_c$. We solve {\sc WeightedLocallyBoundedListColoring} in this case (for arbitrary values of $p$ and for arbitrary vertex weights) by reducing it to a matching instance in a bipartite graph $H$: there is one vertex in $H$ for each vertex $v_i$ in $G=K_n$, and one vertex for each color $c \in \{1, \dots, n\}$. There is an edge between vertex $v_i$ and color $c$ when ($i$) $c \in L(v_i)$, ($ii$) $w(v_i)=W_{h_c c}$ and ($iii$) $v_i \in V_{h_c}$. Due to the rule we used to define the edges of $H$ (by linking vertices of $G$ with ``compatible'' colors only), it is easy to see that we have a solution to the {\sc WeightedLocallyBoundedListColoring} instance if and only if we have a perfect matching in $H$ (a vertex $v_i$ is linked to color $j$ in this matching if and only if it has color $j$ in $G$): therefore, {\sc WeightedLocallyBoundedListColoring} can be solved in polynomial time in this case (using for instance an efficient algorithm computing a maximum flow, since $H$ is bipartite). This yields:

\begin{Th}\label{th:complete-graphs}
{\sc WeightedLocallyBoundedListColoring} can be solved in polynomial time in complete graphs.
\end{Th}

Again, one can associate a profit function $\pi: V \times \{1, \dots, k\} \rightarrow \mathbf Z$ to the vertices of $G$. Then, we can obtain a feasible solution to {\sc WeightedLocallyBoundedListColoring} with maximum profit (if any) by finding a perfect matching of maximum weight in $H$ (or, equivalently, by solving a linear assignment instance optimally), where, for each $i$ and each $j$, the weight of the edge between a vertex $v_i$ and a color $j$ in $H$ is $\pi(v_i,j)$.

\section{Locally bounded list-colorings in split graphs}\label{sect:split-graphs}

In this section, we study the tractability of {\sc WeightedLocallyBoundedListColoring} in split graphs. Recall that a split graph $G$ is a graph whose set of vertices can be partitioned into a vertex set inducing a clique $K$, and a vertex set inducing a stable (or independent) set $S$. For the sake of simplicity, we will denote by $\vert K \vert$ (resp. $\vert S \vert$) the number of vertices of $K$ (resp. of $S$).

From Section \ref{sect:NPC}, {\sc WeightedLocallyBoundedListColoring} remains \textbf{NP}-complete in split graphs: from Theorems \ref{th:NPC-partition} and \ref{th:NPC-3partition}, it is weakly \textbf{NP}-complete when the vertex weights are arbitrary, even if $p=1$, $k=2$, and each vertex can take any color, and strongly \textbf{NP}-complete when $k$ is arbitrary, even if the vertex weights are polynomially bounded, $p=1$, and each vertex can take any color. However, when $k$ is fixed, it can be solved in (pseudo)polynomial time:

\begin{Th}\label{th:split-graphs-P-k-fixed}
{\sc WeightedLocallyBoundedListColoring} can be solved in (pseudo)polynomial time in split graphs when $k=O(1)$.
\end{Th}
\begin{proof}
Since we are looking for a proper vertex coloring, any vertex in $K$ cannot have more than $k-1$ neighbors. Hence, $K$ cannot contain more than $k$ vertices. When $k = O(1)$, one can then enumerate in constant time all the possible colorings of $K$ (for each one, we must check that it is proper and that the color of each vertex $v \in K$ is in $L(v)$). Then, for each such coloring, we remove $K$, update the value of each of the $W_{hc}$'s and the list $L(v)$ for each vertex $v \in S$ accordingly. This way, we obtain an instance in a graph consisting of isolated vertices, and where $k = O(1)$: from Proposition \ref{prop:isolated-vertices-k-fixed}, such an instance can be solved in (pseudo)polynomial time, which concludes the proof.
\end{proof}

The list-coloring problem in complete bipartite graphs (and thus in cographs) is strongly \textbf{NP}-complete from \cite{refJansen}. Actually, in the reduction that is used, one side of the bipartite instances can be replaced by a clique: hence, the proof also holds in split graphs (while the one of Theorem \ref{th:NPC-complete-bipartite} does not, despite the fact that the former one partly inspired the latter one). Adding, for each initial vertex, $(k-1)$ new isolated vertices that can take any color, yields (\textbf{NP}-complete) instances of the equitable list-coloring problem. However, a straightforward adaptation of the proof of Theorem \ref{th:NPC-forests} yields a stronger result:

\begin{Th}\label{th:split-graphs-NPC-unit-weights}
{\sc WeightedLocallyBoundedListColoring} is strongly \textbf{NP}-complete in split graphs, even when $p=1$, $w(v)=1$ for each vertex $v$, and the degree of each vertex in the part $S$ inducing an independent set is one.
\end{Th}
\begin{proof}
Start from the reduction described in the proof of Theorem \ref{th:NPC-forests}, and add a clique $K$ on $3n(n+1)$ vertices, namely all the $u_i$'s and all the $v_i^j$'s. As any two of these vertices do receive different colors in the original reduction, we obtain an equivalent instance. Then, the independent set $S$ contains all the leaves of the $S_i^j$'s, and hence any vertex of $S$ has degree one, which concludes the proof.
\end{proof}

Note, however, that the list-coloring problem is easy in split graphs if every vertex in $S$ has degree at most 1: remove any $v \in S$ with only one possible color, update the possible colors for each $u \in K$ accordingly, use an optimal flow to color the vertices of $K$ appropriately (as in Theorem \ref{th:complete-graphs}), and then give to any other vertex a color not used by its unique neighbor (if such a neighbor exists).

Moreover, the lists of possible colors play a major role in the above reduction, so one may wonder what happens when each vertex can take any color. When $p$ and $k$ are arbitrary, and $w(v)=1$ for each vertex $v$, the following result holds:

\begin{Th}\label{th:split-graphs-NPC-unit-weights-no-list-coloring}
{\sc WeightedLocallyBoundedListColoring} is strongly \textbf{NP}-complete in split graphs, even when $W_{hc} \in \{0,1\}$ for each $h$ and $c$, and each vertex has weight 1 and can take any color.
\end{Th}
\begin{proof}
Take an instance of 3SAT \cite{refGarey}, that consists of $\mu$ clauses of size 3 defined on $\nu$ boolean variables $x_1, \dots, x_{\nu}$. To each variable $x_i$, we associate two vertices $y_i$ and $z_i$. To the $j$th clause, we associate three vertices $u_j,v_j,w_j$, for each $j$. Then, we add edges in the following way, in order to obtain a split graph: the vertices $y_i$ and $z_i$ over all $i$ induce a clique on $2 \nu$ vertices, and the vertices $u_j,v_j,w_j$ over all $j$ induce an independent set on $3 \mu$ vertices. For each $j \in \{1, \dots, \mu\}$, if we let the $j$th clause be $x'_{i_1} \vee x'_{i_2} \vee x'_{i_3}$ (where $x'_{i_l} \in \{x_{i_l}, \bar{x}_{i_l}\}$ for each $l \in \{1,2,3\}$), we also add an edge between $w_j$ and $y_{i_l}$, for each $l \in \{1,2,3\}$. Finally, we let $k = 2 \nu$ and $p=\nu+\mu$, and then define $V_i = \{y_i,z_i\}$ and $W_{ii}=W_{i(\nu+i)}=1$ for each $i \in \{1, \dots, \nu\}$, as well as $V_{\nu+j} = \{u_j,v_j,w_j\}$ and $W_{(\nu+j)\bar{c}_{i_1}}=W_{(\nu+j)\bar{c}_{i_2}}=W_{(\nu+j)\bar{c}_{i_3}}=1$ for each $j \in \{1, \dots, \mu\}$, where, for each $l \in \{1,2,3\}$, $\bar{c}_{i_l}=\nu+i_l$ if $x'_{i_l}=x_{i_l}$, and $\bar{c}_{i_l}=i_l$ otherwise. Note that any other $W_{hc}$ is equal to 0, and that, for each vertex $v$, $w(v)=1$ and $L(v) \in \{1, \dots, 2 \nu\}$.

We can define the following equivalence between the solutions of these two instances: for each $i$, $x_i = true$ if and only if the color of $y_i$ is $i$. Note that, for each $j \in \{1, \dots, \mu\}$, the vertex $w_j$ must take a color $\bar{c}_{i_l}$ for some $l \in \{1,2,3\}$ such that $y_{i_l}$ takes color $c_{i_l}$ (and hence the literal $x'_{i_l}$ is \emph{true}), where $c_{i_l} = i_l$ if $\bar{c}_{i_l}=\nu+i_l$, and $c_{i_l}=\nu+i_l$ otherwise. This concludes the proof.
\end{proof}

Actually, we can even prove a stronger result, and get rid of the assumption that $p$ must be arbitrary, while allowing each vertex to take any color.

\begin{Th}\label{th:split-graphs-NPC-unit-weights-no-list-coloring-p-fixed}
{\sc WeightedLocallyBoundedListColoring} is strongly \textbf{NP}-complete in split graphs, even when $p=1$, $W_{1c} \in \{1,4\}$ for each color $c$, and each vertex has weight 1 and can take any color.
\end{Th}
\begin{proof}
We reduce from the \textbf{NP}-complete problem \textsc{3-DimensionalMatching}, in which we are given three sets $X,Y,Z$ of elements such that $\vert X \vert = \vert Y \vert = \vert Z \vert$, and a list of triples $\mathcal{T} \subseteq X \times Y \times Z$. One wants to decide whether there exists $\mathcal{T'} \subseteq \mathcal{T}$ covering all the elements of $X$, $Y$, and $Z$, and such that $\vert \mathcal{T'} \vert = \vert X \vert$.

Given an instance $I$ of \textsc{3-DimensionalMatching}, we define the following instance $I'$ of {\sc WeightedLocallyBoundedListColoring}. There are $3 \vert X \vert + \vert \mathcal{T} \vert$ vertices, that can take any color, and have weight 1: one vertex $t_h$ for each triple in $\mathcal{T}$, one vertex $x_i$ for each element in $X$, one vertex $y_j$ for each element in $Y$, and one vertex $z_l$ for each element in $Z$. Moreover, there is an edge between each $t_h$ and all the other vertices, \emph{except} $x_{i_h}$, $y_{j_h}$ and $z_{l_h}$, provided that the $h$th triple of $\mathcal{T}$ is composed of the $i_h$th element of $X$, of the $j_h$th element of $Y$, and of the $l_h$th element of $Z$. In particular, we have $\vert K \vert = \vert \mathcal{T} \vert$ and $\vert S \vert = 3 \vert X \vert$, and we define $k = \vert \mathcal{T} \vert$ and $p=1$. Finally, we define the color bounds as follows: $W_{1c} = 4$ for each $c \in \{1, \dots, \vert X \vert\}$, and $W_{1c} = 1$ for each $c \in \{\vert X \vert+1, \dots, k\}$.

Then, it is easy to see that the following equivalence holds: for each $h$, the $h$th triple in $I$ is in $\mathcal{T'}$ if and only if $t_h$ takes a color $c$ with $W_{1c} = 4$ in $I'$. Indeed, in $I'$, each color \emph{must} appear exactly once in $K$, as $\vert K \vert=k$. Moreover, each vertex $t_h \in K$ is non adjacent to exactly three vertices: hence, if $t_h$ takes some color $c$ with $W_{1c} = 4$, then the three vertices non adjacent to it \emph{must} take color $c$. The $\vert X \vert$ such colors must cover all the vertices of $S$, and $\vert X \vert$ vertices of $K$. Each of the $k-\vert X \vert$ other vertices of $K$ takes a color $c'$ with $W_{1c'} = 1$.
\end{proof}

\begin{Cor}\label{cor:split-graphs-NPC-capacitated-coloring}
The capacitated coloring problem is \textbf{NP}-complete in split graphs.
\end{Cor}

Note that, on the contrary, the bounded coloring problem is easy in split graphs \cite{refBodlaenderJansen95}. We end this section by providing three last tractable special cases.

The first one is obtained by assuming that $\vert S \vert = O(1)$. In this case, we enumerate the $O(k^{\vert S \vert})$ potential colorings of $S$ (for each one, we must check that it is proper, and that the color of each vertex $v \in S$ is in $L(v)$). Then, for each such coloring, we remove $S$, update the values $W_{hc}$ and the lists $L(v)$ for all the vertices $v \in K$ accordingly, and we are left with an instance in a complete graph, that we can solve in polynomial time thanks to Theorem \ref{th:complete-graphs}.

The second one is obtained by assuming that the tree-width is $O(1)$ (i.e., that $\vert K \vert = O(1)$) and $w(v)=1$ for each vertex $v$ (without the latter assumption, the problem is \textbf{NP}-complete from Theorems \ref{th:NPC-partition} and \ref{th:NPC-3partition}). In this case, we enumerate the $O(k^{\vert K \vert})$ potential colorings of $K$ (for each one, we must check that it is proper, and that the color of each $u \in K$ is in $L(u)$). Then, for each such coloring, we remove $K$, update the values $W_{hc}$ and the list $L(v)$ of each $v \in S$ accordingly, and we are left with an instance in a set of isolated vertices and where each vertex has weight 1, that we can solve efficiently thanks to Proposition \ref{prop:isolated-vertices-poly}.

The third one complements Corollary \ref{cor:split-graphs-NPC-capacitated-coloring}, and is obtained by assuming that each vertex has weight 1 and can take any color, and that, for all colors $c$ except a constant number $k'$ of them, $W_{hc} = B$ for each $h$, for some common constant $B$ (the assumption $k'=O(1)$ being less restrictive than $k=O(1)$).

Recall that it was proved in \cite{refBodlaenderJansen95} that the bounded coloring problem can be solved in polynomial time in split graphs, by means of a reduction to an intermediate problem, which is then solved using a maximum flow algorithm. The construction we will use to solve this third special case (in which the number $k$ of colors is arbitrary) also solves the bounded coloring problem, with a more direct reduction to a maximum flow problem. Namely, we prove:

\begin{Th}\label{th:split-graphs-poly-max-flow}
{\sc WeightedLocallyBoundedListColoring} can be solved in polynomial time in split graphs when (i) each vertex has weight 1 and can take any color, and (ii) for all colors $c$ except a constant number $k'$ of them (called the \emph{singular} colors), $W_{hc} = B$ for each $h$, for some constant $B$.
\end{Th}
\begin{proof}
Let us consider such an instance in a split graph $G$, and let $u_1,\dots,u_{\vert K \vert}$ be the vertices of $K$, and $v_1,\dots,v_{\vert S \vert}$ the vertices of $S$. We start by ``guessing'' the singular colors that will appear in $K$, and for each such color which (unique) vertex of $K$ uses it. This can be done in $O\left((\vert K \vert +1)^{k'}\right)$ time by a brute-force enumeration. Any other vertex of $K$ will take an arbitrary (but different for each of them) color $c$ such that $W_{hc} = B$ for each $h$, and some singular colors may appear only in $S$. Also, for each $h$ and each $c$, we set $W'_{hc}=W_{hc}-w(u_i)=W_{hc}-1$ if some vertex $u_i \in K \cap V_h$ takes color $c$, and $W'_{hc}=W_{hc}$ otherwise.

Since $k \geq \vert K \vert$ (otherwise, there is no solution), we always end up with a proper coloring of $K$. Then, we define for each $v_j \in S$ a list of possible colors $L(v_j)=\{1,\dots,k\} \setminus \{c: \text{ there exists $u_i \in K$ adjacent to $v_j$ that takes color $c$}\}$. Now, we can remove the vertices of $K$, and we are left with an instance of {\sc WeightedLocallyBoundedListColoring} in a graph consisting of isolated vertices of weight 1 (i.e., in the subgraph of $G$ induced by the vertices of $S$). The color bounds are the $W'_{hc}$'s, and we can solve this instance in polynomial time (as a maximum matching problem in a bipartite graph, and hence as a maximum flow problem) thanks to Proposition \ref{prop:isolated-vertices-poly}, which concludes the proof.
\end{proof}

Again, note that this generalizes the polynomial-time solvability of the case where $p=1$ and $k'=0$, i.e., of the bounded coloring problem in split graphs.

Moreover, without the assumption that each vertex can take any color, the problem becomes hard, even when $p=1$. Indeed, take the reduction from Theorem \ref{th:split-graphs-NPC-unit-weights-no-list-coloring-p-fixed}, and, for each color $c$ such that $W_{1c}=1$, add three isolated vertices that can only take color $c$. This way, we obtain (\textbf{NP}-complete) instances where $W_{1c}=4$ for each color $c$, and hence where $k'=0$. Similarly, adapting the proof of Theorem \ref{th:split-graphs-NPC-unit-weights-no-list-coloring} shows that, without the above assumption, the problem is \textbf{NP}-complete when $p$ is arbitrary and $W_{hc}=1$ for each $h$ and each $c$.

However, if each vertex in the clique part $K$ of $G$ has both a list of possible colors and an arbitrary weight, then the proof of Theorem \ref{th:split-graphs-poly-max-flow} can easily be adapted to show that the problem remains tractable, provided that each vertex of $S$ has weight 1 and is allowed to take any color. Indeed, after the ``guess'' part, we can give appropriate non singular colors to the remaining uncolored vertices of $K$ by solving another matching instance (and not in a greedy way anymore), in a way similar to the one used in the proof of Theorem \ref{th:complete-graphs}.

\section{Extension to edge colorings} \label{sect:edge-colorings}

In this last section, we extend to edge colorings all our previous results concerning (vertex) $k$-colorings in cographs, split graphs, and graphs of bounded tree-width. An edge coloring of a graph is an assignment of colors (integers) to its edges, in such a way that any two edges sharing a vertex take different colors. The problem we consider in this section is thus the following:\\

\noindent {\sc WeightedLocallyBoundedListEdgeColoring}\\
\textit{Instance}: A graph $G=(V,E)$, a weight function $w: E \rightarrow \mathbf N^*$, a partition $E_1, \dots, E_p$ of the edge set $E$, a list of $pk$ integral bounds $W_{11}, \dots, W_{1k}, \dots, W_{p1},$ $\dots, W_{pk}$ such that $\sum_{j=1}^{k} W_{ij} = w(E_i)$ for each $i \in \{1, \dots, p\}$ (where $w(E_i)=\sum_{e:e \in E_i} w(e)$), and, for each edge $e$, a list of possible colors $L(e) \subseteq \{1, \dots, k\}$.\\
\textit{Question}: Decide whether there exists an edge coloring of $G$ such that, for each $i \in \{1, \dots, p\}$ and for each $j \in \{1, \dots, k\}$, the total weight of edges having color $j$ in $E_i$ is $W_{ij}$, and the color of $e$ belongs to $L(e)$ for each edge $e$.

\subsection{\textbf{NP}-completeness proofs for edge colorings}\label{sect:NPC-edge-coloring}

We can prove that this problem is \textbf{NP}-complete by using a reduction from {\sc Partition}, similar to the one in Theorem \ref{th:NPC-partition} (so we have $p=1$, $k=2$ and $W_{11}=W_{12}=B$). Each isolated vertex $v_i$ becomes an isolated edge $e_i$, and we define $w(e_i)=a_i$ and $L(e_i)=\{1,2\}$ for each $i \in \{1, \dots, n\}$. Again, we can make this instance connected, and obtain a chain, by setting $k=3$ and adding new edges of color 3 between the $e_i$'s. However, the problem is trivial when $k=2$ and $G$ is connected, since $G$ has maximum degree 2 in this case, and hence $G$ is either a path or an (even) cycle, and has only two possible edge-colorings.

The main goal of this section is to prove that {\sc WeightedLocallyBoundedListEdgeColoring} is tractable in the graphs considered in the previous sections, under the same assumptions as its vertex coloring counterpart. However, we shall first justify that, as in the case of vertex colorings, dropping any of these assumptions leads to \textbf{NP}-completeness, without looking further into other possible special cases, for the sake of brevity. The previous paragraph settles the case of arbitrary edge weights, so let us look at the other cases.

When the input graph is neither a cograph nor a split graph, and has an unbounded tree-width, {\sc WeightedLocallyBoundedListEdgeColoring} is \textbf{NP}-complete even when $k=3$, $p=1$, and we have both $L(e) = \{1, 2, 3\}$ and $w(e)=1$ for each edge $e$. Indeed, deciding whether the edges of a graph can be properly colored using three colors is known to be \textbf{NP}-complete \cite{refHolyer81}.

When $p$ is arbitrary, we need an adaptation of the proof of Theorem \ref{th:NPC-forest-1in3SAT}:

\begin{Th}\label{th:NPC-edge-coloring-p-arbitrary}
{\sc WeightedLocallyBoundedListEdgeColoring} is strongly \textbf{NP}-complete in graphs where each connected component is either a cycle of length four or a single edge, even when $k=2$, $w(e)=1$ for each edge $e$, and each edge can take any color.
\end{Th}
\begin{proof}
Given an instance of {\sc MonotoneOneInThreeSAT}, we construct such a graph. More precisely, for each variable $x_i$, there are $occ(i)$ isolated edges (denoted by $e^j_i$) and $occ(i)$ vertex-disjoint cycles of length four, where $occ(i)$ is the number of occurrences of $x_i$ in the set of clauses. For each $i \in \{1, \dots, \nu\}$ and each $j \in \{1, \dots, occ(i)\}$, let us denote by $a^j_i,b^j_i,c^j_i,d^j_i$ the four edges of the $j$th of these cycles, and assume for instance that $a^j_i$ and $c^j_i$ are vertex-disjoint, and that so are $b^j_i$ and $d^j_i$. Moreover, each edge has weight 1 and can take any of the two colors. Intuitively, in the $j$th cycle associated with variable $x_i$, the edge $b^j_i$ will be used to interact with the $j+1$th one, the edge $c^j_i$ will be used to interact with the $j-1$th one, the edge $d^j_i$ will only be useful to obtain a cograph, and the color of the edge $a^j_i$ will reflect the value of $x_i$ (\emph{true} or \emph{false}).

Let us define formally the partition that we use, and the associated color bounds. For each $i \in \{1, \dots, \nu\}$ and each $j \in \{1, \dots, occ(i)\}$, we set $W_{h1}=W_{h2}=1$, as well as $V_h = \{b^j_i,c^{j+1}_i\}$ if $j < occ(i)$ and $V_h = \{b^{occ(i)}_i,c^1_i\}$ otherwise, where $h = 2 \sum_{q=1}^{i-1} occ(q) + j$. Moreover, for each $i \in \{1, \dots, \nu\}$ and each $j \in \{1, \dots, occ(i)\}$, we set $V_{h'} = \{d^j_i,e^j_i\}$ and $W_{h'1}=W_{h'2}=1$, where $h' = 2 \sum_{q=1}^{i-1} occ(q) + occ(i) + j$. Note that $k=2$, and hence, in the $j$th cycle of length four associated with variable $x_i$, edges $a^j_i$ and $c^j_i$ must take the same color, and edges $b^j_i$ and $d^j_i$ must take the same color. For each $i$, by the way we defined the $V_h$'s and $V_{h'}$'s so far, \emph{all} the $a^j_i$'s must take the same color (1 or 2). We end the reduction by defining one set of the partition for each clause: if the $l$th clause consists of the $j_1$th occurrence of variable $x_{i_1}$, of the $j_2$th occurrence of variable $x_{i_2}$, and of the $j_3$th occurrence of variable $x_{i_3}$, then we define $V_{h''} = \{a^{j_1}_{i_1}, a^{j_2}_{i_2}, a^{j_3}_{i_3}\}$, $W_{h''1}=1$ and $W_{h''2}=2$, where $h'' = 2 \sum_{q=1}^{\nu} occ(q) + l$. This implies that $p = 2 \sum_{q=1}^{\nu} occ(q) + \mu$, and this yields the following equivalence between the solutions of the two instances: for each $i$, $x_i = true$ if and only if $a^j_i$ has color 1 for each $j$, which concludes the proof.
\end{proof}

When the number of colors is arbitrary, {\sc WeightedLocallyBoundedListEdgeColoring} is strongly \textbf{NP}-complete, as shown by an easy reduction from {\sc $3-$Partition}, similar to the one in Theorem \ref{th:NPC-3partition} (so we have $p=1$, $k=n$, and $W_{1c} = B$ for each $c \in \{1, \dots, n\}$). Each isolated vertex $v_i$ becomes an isolated edge $e_i$, and we define $w(e_i)=a_i$ and $L(e_i)=\{1, \dots, k\}$ for each $i \in \{1, \dots, 3n\}$.


Note that the graphs used in the reductions described above consist of vertex-disjoint cycles of length four (i.e., complete bipartite graphs with two vertices on each side) and/or of isolated edges. In other words, they are both cographs and graphs of tree-width at most 2. Moreover, we can easily turn the instances used in the case where the number $k$ of colors is arbitrary into split graphs.

To see it, start from the above reduction, but define $k=n+\frac{3n(3n-1)}{2}$ instead of $k=n$. Then, choose one endpoint of each $e_i$, $i \in \{1, \dots, 3n\}$, and add $\frac{3n(3n-1)}{2}$ edges (i.e., a clique) between these $3n$ vertices (this yields a split graph). These edges will be denoted by $f_1, f_2, \dots$, and can take any color. We define $w(f_j)=1$ and $W_{1(n+j)} = 1$ for each $j \in \{1, \dots, \frac{3n(3n-1)}{2}\}$. As we can assume without loss of generality that $a_i > 1$ for each $i$ (simply multiply $B$ and each $a_i$ by two in the original instance of {\sc $3-$Partition}, in order to obtain such an equivalent instance), this implies that each color $c \geq n+1$ will be used by one and only one edge $f_j$, which shows the equivalence between the two reductions.

Hence, from now on, we shall assume that $k$ is fixed, as otherwise {\sc WeightedLocallyBoundedListEdgeColoring} is \textbf{NP}-complete. In the following subsection, we give efficient algorithms to solve this problem when $k = O(1)$.

\subsection{Efficient algorithms for edge colorings when $k = O(1)$}\label{sect:algo-edge-coloring}

We start by considering graphs of bounded tree-width. We assume that such a graph $G$ is given, together with a tree-decomposition $T=(I,F)$ of $G$ having minimum width, and that $p$ is fixed (as otherwise {\sc WeightedLocallyBoundedListEdgeColoring} is \textbf{NP}-complete from Theorem \ref{th:NPC-edge-coloring-p-arbitrary}). We then define the following function, before showing how to compute it for each node type:
\begin{itemize}
\item $g(i, c_i, \omega_{11}, \dots, \omega_{1k}, \dots, \omega_{p1}, \dots, \omega_{pk})$ = \emph{true} if there is a list-edge-coloring of $G[T_i]$ where each edge $e \in Y_i$ has color $c_i(e)$ and where the total weight of edges of $G[T_i]$ having color $c$ in $E_h$ is $\omega_{hc}$, and \emph{false} otherwise.
\end{itemize}

\paragraph*{If $i$ is a forget node.}
Let $j$ denote its child such that $X_j=X_i \cup \{v\}$:
\begin{equation*}
g(i, c_i, \omega_{11}, \dots, \omega_{pk})=\bigvee_{c_j:(c_j(e) = c_i(e)~\forall e \in Y_i) \wedge (c_j(uv)\in L(uv)~\forall u \in X_i:uv \in Y_j)} g(j, c_j, \omega_{11}, \dots, \omega_{pk})
\end{equation*}

\paragraph*{If $i$ is an introduce node.}
Let $j$ denote its child such that $X_j=X_i \setminus \{v\}$. We denote by $e_1,\dots,e_d$ the edges incident to $v$ in $Y_i$ (i.e., the edges in $Y_i \setminus Y_j$), and assume that $e_h \in E_{\phi(h)}$ for each $h \in \{1,\dots,d\}$.
\begin{equation*}
g(i, c_i, \omega_{11}, \dots, \omega_{pk}) = \mathcal{E} \bigwedge g(j, c_j, \omega_{11}, \dots, \omega_{\phi(1) c_i(e_1)}-w(e_1), \dots, \omega_{\phi(d) c_i(e_d)}-w(e_d), \dots, \omega_{pk})\\
\end{equation*}
where $\mathcal{E} = (c_j(e)=c_i(e)~\forall e \in Y_j) \bigwedge (c_i(e) \in L(e)~\forall e \in \{e_1, \dots, e_d\}) \bigwedge (c_i(e) \neq c_i(f)~\forall e \in \{e_1, \dots, e_d\}, \forall f \in Y_i : e \neq f \text{ and both edges share a common vertex})\\\bigwedge (\omega_{\phi(h) c_i(e_h)} \geq w(e_h)~\forall h \in \{1,\dots,d\})$.

\paragraph*{If $i$ is a join node.}
Let $j$ and $l$ denote its two children, and let $w^i_{hc}$ be the total weight of edges $e \in Y_i \cap E_h$ such that $c_i(e)=c$, for each $h$ and $c$.
\begin{equation*}
  g(i, c_i, \omega_{11}, \dots, \omega_{pk})=\bigvee_{(q_{11},\dots,q_{pk}) \in Q^i}\big(g(j, c_i, \omega_{11}+w^i_{11}-q_{11}, \dots, \omega_{pk}+w^i_{pk}-q_{pk}) \wedge g(l, c_i, q_{11}, \dots, q_{pk})\big)
\end{equation*}
where $Q^i=\{(q_{11},\dots,q_{pk}):w^i_{hc} \leq q_{hc}\leq \omega_{hc}~\forall h~\forall c\}$.

\paragraph*{If $i$ is a leaf node.}
In this case, we just have to check that the coloring function $c_i$ provides a valid locally bounded list-edge-coloring of $G[X_i]$.
\begin{equation*}
g(i, c_i, \omega_{11}, \dots, \omega_{pk})= \mathcal{E} \bigwedge (w(\{e\in Y_i\cap E_h:c_i(e)=c\}) = \omega_{hc}~\forall h~\forall c)
\end{equation*}
where $\mathcal{E}=(c_i(e) \in L(e)~\forall e \in Y_i) \bigwedge (c_i(u_1v) \neq c_i(u_2v)~\forall u_1,u_2,v \in X_i : u_1v \in Y_i, u_2v \in Y_i, u_1 \neq u_2)$.

\paragraph*{Root value.} Again, the answer \emph{true} or \emph{false} for the initial instance is obtained at the root $r$ by computing the following value:
\begin{equation*}
\bigvee_{c_r:c_r(e)\in L(e)~\forall e \in Y_r} g(r, c_r, W_{11}, \dots, W_{pk})
\end{equation*}

\noindent The proof of correctness is similar to the one given in Lemma \ref{lem:algo-btw}, and the analysis of the running time is similar to the one given just before stating Theorem \ref{th:algo-btw}. Therefore, they are both omitted. This yields the following result:

\begin{Th}
In graphs of bounded tree-width, {\sc WeightedLocallyBoundedListEdgeColoring} can be solved:
\begin{itemize}
\item in pseudopolynomial time when $p$ and $k$ are fixed,
\item in polynomial time when (i) $p$ and $k$ are fixed, and (ii) all edge weights are polynomially bounded.
\end{itemize}
\end{Th}

Again, one can associate a profit function to the edges of $G$, and then modify slightly this dynamic programming algorithm, in order to compute a feasible weighted and locally bounded list-edge-coloring of maximum profit (if any).

We then turn our attention towards cographs, and assume that $p$ is fixed (as otherwise an \textbf{NP}-completeness result holds, from Theorem \ref{th:NPC-edge-coloring-p-arbitrary}). Since we are looking for a proper edge coloring, the maximum degree of such a graph $G$ is bounded by the number $k$ of colors. As a cograph contains no induced path on four vertices, the diameter of any of its connected components is at most two. In other words, any vertex $v$ has at most $k$ neighbors, each of its neighbors has at most $k-1$ other neighbors, and there can be no other vertex in the connected component $v$ belongs to. This implies that any connected component contains at most $1+k+k(k-1)=O(k^2)$ vertices, and hence $O(k^4)$ edges. Therefore, when $k=O(1)$, each connected component has $O(1)$ possible edge colorings. For each one of them, one can easily check whether it defines a proper list-coloring of the component. Moreover, as $p$ and $k$ are fixed, these colorings can be combined one by one, by using an efficient dynamic programming algorithm similar to the one outlined at the beginning of Section \ref{sect:NPC}, that needs to store a (pseudo)polynomial amount of information, namely $O((\max_{h,c} W_{hc})^{kp})$.

Finally, let us study the case of split graphs. Again, since we are looking for a proper edge coloring, the maximum degree of such a graph $G$ is bounded by the number $k$ of colors, and isolated vertices play no role in this case (they can be removed). Hence, each vertex in $K$, the clique part of $G$, cannot have more than $k$ neighbors, which implies that $K$ cannot contain more than $k+1$ vertices. As any edge of $G$ is incident to one of these vertices, and as each of these vertices cannot have more than $k$ incident edges, the graph $G$ contains $O(k^2)$ edges. Therefore, when $k = O(1)$ (and hence $O(k^2) = O(1)$), {\sc WeightedLocallyBoundedListEdgeColoring} is trivial in split graphs.

\section{Conclusion and open problems}\label{sect:conclusion}

We conclude by giving an overview of all our results concerning {\sc WeightedLocallyBoundedListColoring} in cographs, split graphs, and graphs of bounded tree-width. The three following tables summarize these results, each one being devoted to one of these classes of graphs (here, \emph{NPC} means \emph{\textbf{NP}-complete}, \emph{P} means \emph{solvable in polynomial time}, "?" means \emph{arbitrary}, and $tw(G)$ denotes the tree-width of the input graph $G$).

\bigskip
\centerline{
\begin{tabular}{|l|l|l|l|l|l|}
\hline
\multicolumn{6}{|c|}{\emph{GRAPHS OF BOUNDED TREE-WIDTH}}\\
\hline
\hline
$tw(G)$ & $w(v)$ & $k$ & $p$ & List-coloring & Result\\
\hline
\hline
$tw(G) = 0$ & ? & $k=2$ & $p=1$ & NO & \emph{NPC} (in the weak sense) from Theorem \ref{th:NPC-partition}\\
\hline
$tw(G) = 0$ & $w(v)=poly(n)$ $\forall v \in V$ & ? & $p=1$ & NO & \emph{NPC} from Theorem \ref{th:NPC-3partition}\\
\hline
$tw(G) = 0$ & $w(v)=1$ $\forall v \in V$ & ? & ? & YES & \emph{P} from Proposition \ref{prop:isolated-vertices-poly}\\
\hline
$tw(G) = 0$ & $w(v)=poly(n)$ $\forall v \in V$ & $k=O(1)$ & ? & YES & \emph{P} from Proposition \ref{prop:isolated-vertices-k-fixed}\\
\hline
$tw(G) = 1$ & $w(v)=1$ $\forall v \in V$ & ? & $p=1$ & YES & \emph{NPC} from \cite{refBodlaender05,refGravier99} and Theorems \ref{th:NPC-forests} and \ref{th:NPC-linear-forests}\\
\hline
$tw(G) = 1$ & $w(v)=1$ $\forall v \in V$ & $k=2$ & ? & NO & \emph{NPC} from Theorems \ref{th:NPC-forest-1in3SAT} and \ref{th:NPC-linear-forest-1in3SAT}\\
\hline
$tw(G) = O(1)$ & $w(v)=poly(n)$ $\forall v \in V$ & $k=O(1)$ & $p=O(1)$ & YES & \emph{P} from Theorem \ref{th:algo-btw}\\
\hline
\end{tabular}}

\bigskip
\centerline{
\begin{tabular}{|l|l|l|l|l|}
\hline
\multicolumn{5}{|c|}{\emph{COGRAPHS}}\\
\hline
\hline
$w(v)$ & $k$ & $p$ & List-coloring & Result\\
\hline
\hline
? & $k=2$ & $p=1$ & NO & \emph{NPC} (in the weak sense) from Theorem \ref{th:NPC-partition}\\
\hline
$w(v)=poly(n)$ $\forall v \in V$ & ? & $p=1$ & NO & \emph{NPC} from Theorem \ref{th:NPC-3partition}\\
\hline
$w(v)=1$ $\forall v \in V$ & ? & $p=1$ & YES & \emph{NPC} from Corollary \ref{th:NPC-forests-cographs} and \cite{refBodlaenderJansen95} (without list-coloring)\\
\hline
$w(v)=1$ $\forall v \in V$ & $k=2$ & ? & NO & \emph{NPC} from Corollary \ref{th:NPC-forests-cographs-1in3SAT}\\
\hline
$w(v)=poly(n)$ $\forall v \in V$ & $k=O(1)$ & $p=O(1)$ & YES & \emph{P} from Theorem \ref{th:algo-cographs}\\
\hline
$w(v)=1$ $\forall v \in V$ & ? & $p=1$ & YES & \emph{NPC} in complete bipartite graphs from Theorem \ref{th:NPC-complete-bipartite}\\
\hline
$w(v)=poly(n)$ $\forall v \in V$ & $k=O(1)$ & ? & YES & \emph{P} in complete bipartite graphs from Theorem \ref{th:complete-bipartite-graphs-k-fixe}\\
\hline
? & ? & ? & YES & \emph{P} in complete graphs from Theorem \ref{th:complete-graphs}\\
\hline
\end{tabular}}

\bigskip
\centerline{
\begin{tabular}{|l|l|l|l|l|}
\hline
\multicolumn{5}{|c|}{\emph{SPLIT GRAPHS}}\\
\hline
\hline
$w(v)$ & $k$ & $p$ & List-coloring & Result\\
\hline
\hline
? & $k=2$ & $p=1$ & NO & \emph{NPC} (in the weak sense) from Theorem \ref{th:NPC-partition}\\
\hline
$w(v)=poly(n)$ $\forall v \in V$ & ? & $p=1$ & NO & \emph{NPC} from Theorem \ref{th:NPC-3partition}\\
\hline
$w(v)=1$ $\forall v \in V$ & ? & $p=1$ & YES & \emph{NPC} from \cite{refJansen} and Theorem \ref{th:split-graphs-NPC-unit-weights}\\
\hline
$w(v)=1$ $\forall v \in V$ & ? & ? & NO & \emph{NPC} from Theorem \ref{th:split-graphs-NPC-unit-weights-no-list-coloring} (even if $W_{hc} \in \{0,1\}$ for each $h$ and $c$)\\
\hline
$w(v)=1$ $\forall v \in V$ & ? & $p=1$ & NO & \emph{NPC} from Theorem \ref{th:split-graphs-NPC-unit-weights-no-list-coloring-p-fixed}\\
\hline
$w(v)=1$ $\forall v \in V$ & ? & ? & NO & \emph{P} if $O(1)$ singular colors from Theorem \ref{th:split-graphs-poly-max-flow}\\
\hline
$w(v)=poly(n)$ $\forall v \in V$ & $k=O(1)$ & ? & YES & \emph{P} from Theorem \ref{th:split-graphs-P-k-fixed}\\
\hline
\end{tabular}}

~\\~\\
The following table summarizes results concerning arbitrary graphs:

\bigskip
\centerline{
\begin{tabular}{|l|l|l|l|l|}
\hline
$w(v)$ & $k$ & $p$ & List-coloring & Result\\
\hline
\hline
$w(v)=poly(n)$ $\forall v \in V$ & $k=2$ & $p=O(1)$ & YES & \emph{P} from Section \ref{sect:NPC}\\
\hline
$w(v)=1$ $\forall v \in V$ & $k=3$ & $p=1$ & NO & \emph{NPC} (since 3-coloring is)\\
\hline
\end{tabular}}

~\\

Again, we highlight that, in Section \ref{sect:edge-colorings}, similar results for edge colorings (including some specific hardness proofs, such as the one of Theorem \ref{th:NPC-edge-coloring-p-arbitrary}) are provided in cographs, split graphs, and graphs of bounded tree-width.

~\\


One may actually wonder how {\sc WeightedLocallyBoundedListColoring} behaves in bipartite graphs, which we did not mention yet. The instances used in the proofs of Theorems \ref{th:NPC-partition} to \ref{th:NPC-linear-forest-1in3SAT} being such graphs, the only relevant question would be: is this problem polynomial-time solvable in bipartite graphs when $k = O(1)$, $p = O(1)$, and all the vertex weights are polynomially bounded? Unfortunately, the answer is \emph{no}. Indeed, Bodlaender and Jansen proved in \cite[Theorem 5.1]{refBodlaenderJansen95} that the \emph{equitable} coloring problem is strongly \textbf{NP}-complete in bipartite graphs (this theorem was stated only for the bounded coloring problem, but the proof actually holds in this case as well, without any modification), even with only $k=3$ colors, and so is {\sc WeightedLocallyBoundedListColoring} when $p=1$, $k=3$, and any vertex has weight 1 and can take any color.

One may also be interested in FPT algorithms for {\sc WeightedLocallyBoundedListColoring} with polynomially bounded vertex weights. Unfortunately, the results in Theorems \ref{th:algo-btw}, \ref{th:algo-cographs} and \ref{th:split-graphs-P-k-fixed} are best possible, as this problem is \textbf{W[1]}-hard with respect to $k$ in graphs of tree-width 0 (which are also cographs and split graphs) if $p=1$, even with no list-coloring. To see it, take any instance of \emph{unary} bin packing with $l$ bins, which is \textbf{W[1]}-hard with respect to $l$ \cite{refJansenMarx13}, and build an equivalent instance of the target problem: each of the $n$ items becomes a vertex (of the same weight $w_i$), the $k=l$ color bounds equal the common bin capacity $B$, and we add $kB-\sum_{i=1}^n{w_i}$ more isolated vertices of weight 1.

Finally, the main open question we would like to mention as worth studying is a direct consequence of our results in Sections \ref{sect:NPC} and \ref{sect:btw}. Indeed, Theorems \ref{th:NPC-partition} to \ref{th:algo-btw} leave as open the case where $p=1$, the tree-width is $O(1)$, and each vertex has weight 1 and can take any color. In other words, what is the complexity of the capacitated coloring problem (as defined in \cite{refBonomo11}) in graphs of bounded tree-width when $k$ is arbitrary? Bodlaender and Fomin managed to prove in \cite{refBodlaender05} that the bounded coloring problem is tractable in this case, a question that was open for a long time before they finally settled it. The capacitated coloring problem is more general, and should probably be harder (even in this case), but no one has been able to prove it so far. Actually, there is a reduction from the capacitated coloring problem to the bounded (or equitable) coloring problem: to see it, take any instance of the former problem with color bounds $n_1 \geq n_2 \geq \ldots \geq n_k$, add isolated vertices so that the total number of vertices reaches $\sum_{i=1}^k n_i$, and then add $k$ sets of vertices inducing independent sets, in such a way that the $i$th of these sets contains $n_1-n_i+1$ isolated vertices and any two vertices in any two different sets are linked by an edge. We obtain an instance of the latter problem by setting the common color bound to $n_1+1$, as it can be checked that \emph{all} the vertices in the $i$th set \emph{must} take color $i$. However, in general, this reduction preserves neither the property of having bounded tree-width nor the one of being a split graph (although it does preserve the one of being a cograph). Also recall that the capacitated coloring problem is \textbf{NP}-complete in cographs and split graphs when $k$ is arbitrary, from \cite{refBodlaenderJansen95} and Corollary \ref{cor:split-graphs-NPC-capacitated-coloring} respectively.

\end{document}